\documentclass[10pt,twocolumn,twoside]{IEEEtran}

\ifCLASSINFOpdf
\else
\fi
\usepackage{url}
\usepackage{amssymb}
\usepackage{amsfonts}
\usepackage{amsmath}
\usepackage{amsmath,amssymb,verbatim}
\usepackage{graphicx}
\usepackage{float}
\usepackage{algorithm}
\usepackage{algorithmic}
\usepackage{epsfig}
\usepackage{subfig}%
\usepackage{url}
\usepackage{color}
\usepackage{delarray}

\usepackage{amsthm}
\usepackage{comment}
\usepackage{ifthen}
\let\labelindent\relax
\usepackage{enumitem}
\usepackage{grffile}

\newboolean{Longversion} \setboolean{Longversion}{true}

\newtheorem{theorem}{Theorem}
\newtheorem{observation}{Observation}
\newtheorem{lemma}{Lemma}

\newtheorem{corollary}{Corollary}

\DeclareGraphicsExtensions{%
     .png,.PNG,%
     .pdf,.PDF,%
    .jpg,.mps,.jpeg,.jbig2,.jb2,.JPG,.JPEG,.JBIG2,.JB2}

\begin{document}
\title{Designing virus-resistant, high-performance \\networks: a game-formation approach}
\author{Stojan Trajanovski,~\IEEEmembership{Member,~IEEE,}
	 Fernando A. Kuipers,~\IEEEmembership{Senior Member,~IEEE,}
        Yezekael Hayel,~\IEEEmembership{Senior Member,~IEEE,}
       Eitan Altman,~\IEEEmembership{Fellow,~IEEE,}
        and Piet Van Mieghem~\IEEEmembership{Member,~IEEE}
        
\thanks{S. Trajanovski is now with the Philips Research, Eindhoven, The Netherlands. The research was done while the first author was with Delft University of Technology (stojan.trajanovski@philips.com).}
\thanks{F. A. Kuipers and P. Van Mieghem are with Delft University of Technology, Faculty of Electrical
Engineering, Mathematics and Computer Science, P.O. Box 5031, 2600 GA Delft,
The Netherlands (\{F.A.Kuipers, P.F.A.VanMieghem\}@tudelft.nl).}
\thanks{Y. Hayel is with University of Avignon, Avignon, France (yezekael.hayel@univ-avignon.fr).}
\thanks{E. Altman is with University of Cote d'Azur, INRIA, BP95, 06902 Sophia Antipolis, France. He is also member of
LINCS, 23 Ave. d'Italie, 75013 Paris, France and an associate member of LIA, University of Avignon, France (eitan.altman@inria.fr).}
\thanks{This article builds on our work~\cite{CDC2015_GameFormationVirusSpread} published in Proc. of IEEE CDC 2015, Osaka, Japan.}
}

\maketitle

\begin{abstract}Designing an optimal network topology while balancing multiple, possibly conflicting objectives like cost, performance, and resiliency to viruses is a challenging endeavor, let alone in the case of decentralized network formation. We therefore propose a game-formation technique where each player aims to minimize its cost in installing links, the probability of being infected by a virus and the sum of hopcounts on its shortest paths to all other nodes.

In this article, we (1) determine the Nash Equilibria and the Price of Anarchy for our novel network formation game, (2) demonstrate that the Price of Anarchy (PoA) is usually low, which suggests that (near-)optimal topologies can be formed in a decentralized way, and (3) give suggestions for practitioners for those cases where the PoA is high and some centralized control/incentives are advisable.
\end{abstract}
\vspace{2em}


\section{Introduction}\label{sec:introduction}
Designing communication and computer networks are complex processes in which careful trade-offs have to be made with respect to performance, resiliency/security and cost investments. For instance, if a host in a computer network wants to route traffic to multiple other hosts, it could directly connect to those other hosts, in this way increasing its expenses in installing and maintaining these connections and at the same time also becoming more susceptible to viruses from those other hosts. In return, it would obtain a better and faster performance with minimum delays, compared to when it would have used intermediate hosts as relays. Although in this example, both installation costs and risk to viruses are increasing, they are linearly independent and they do not necessarily optimize together. Indeed, reducing the number of direct connections would reduce the cost and the host would be less vulnerable to viruses. However, even when being connected to a few high-degree nodes with direct connections, the host would still be seriously imposed to a virus.

In practice, hosts often are autonomous, act independently and do not coordinate as in P2P networks~\cite{neglia07}, peering relations between Autonomous Systems~\cite{Corbo2005}, overlay networks~\cite{Chun2004}, wireless~\cite{Shakkottai2007,Krunz2014,mackenzie2006game} and mobile~\cite{Iosifidis2014} networks, resource sharing in VoIP networks~\cite{Watanabe2008}, social networks~\cite{Jiang2013,Marbach2011} or the Internet~\cite{Fabrikant:2003:NCG:872035.872088}. Their aim is to optimize their own utility functions, which are not necessarily in accordance to the global optimum. To study global network formation under autonomous actors, the network formation game (NFG) framework~\cite{Aumann88} has been proposed. However, resilience and notably virus protection have not been taken into account in that NFG context, even though their importance is undisputed. In this paper, we therefore take the NFG framework one step further by including performance and virus protection as key ingredients. Virus propagation will be modeled by the Susceptible-Infected-Susceptible (SIS) model~\cite{Omic09} and we will evaluate the effect of uncoordinated autonomous hosts \emph{versus} the optimal network topology via standard game-theoretic concepts, such as Nash Equilibria and the Prices of Anarchy and Stability. 

Our network formation game is called the \emph{Virus Spread-Performance-Cost} (\textsc{VSPC}) game. Each node (i.e., autonomous player) attempts to minimize both the cost and infection probability, while still being able to route traffic to all the other nodes in a small number of hops. When the hopcount performance metric is irrelevant, the game is driven by the cost and virus objectives; a scenario we studied in~\cite{CDC2015_GameFormationVirusSpread}. That particular scenario resulted in sparse graphs, which may not always represent real-world networks, but it helped to understand the process of virus spread better. In this paper, we generalize those results by also including the hopcount performance metric. The probability of the node being infected and the hopcounts to the other nodes change in a different direction, for example adding a link reduces the former, but increases the latter. Therefore, there is a tradeoff in the number of added links and how these new links are best added. Moreover, the two metrics are linearly independent and closed-form expressions do not exist, which makes the problem complex. Finally, the inclusion of the hopcount allows us to better capture realistic networks. In particular, our main contributions are the following:
\begin{itemize}
\item We provide a complete characterization of the various relevant parameter settings and their impact on the formation of the topologies. 
\item We show that depending on the input, the Nash Equilibria may vary from tree graphs, via graphs of different diameters, to complete graphs. 
\item We demonstrate, both via theory and simulations, that the Price of Anarchy (PoA) is small in most of the cases, which implies that (near-)optimal topologies can be formed in a decentralized non-cooperative manner. We will also identify for which scenarios the PoA may be high. In those cases a central point of control would be desirable to limit/steer the players' decisions.
\end{itemize}

This paper is organized as follows: The SIS-virus spread model and the network-formation game model are introduced in Section~\ref{sec:Models}. The \emph{Virus Spread-Performance-Cost} (\textsc{VSPC}) game formation is analyzed in Section~\ref{sec:networkFormationVirusSpreadPerformance}. Related work on game formation and protection against viruses is discussed in Section~\ref{sec:RelatedWork}. The conclusion and directions for future work are provided in Section~\ref{sec:Conclusion}.

\section{Models and problem statements}\label{sec:Models}
\subsection{Virus-spread model}
The spread of viruses in communication and computer networks can be described, using virus-spread epidemic models~\cite{Omic09,Chakrabarti2008,Ganesh2005}. We consider the Susceptible-Infected-Susceptible (SIS) NIMFA model~\cite{Omic09,VanMieghem2011},
\begin{align}  
\frac{dv_{i}\left( t\right) }{dt}&=\beta\left( 1-v_{i}\left( t\right)
\right) \sum\limits_{j=1}^{N}a_{ij} v_{j}\left( t\right) -\delta v_{i}\left(
t\right) \label{N_intertwined}
\end{align}
where $N$ is the number of network nodes and $v_{i}(t)$ is the probability of node $i$
being infected at time $t$, for all $ i \in \{1,2,\ldots, N\}$. If a link is present between nodes $i$ and $j$, then $a_{ij}=1$, otherwise $a_{ij}=0$. In (\ref{N_intertwined}), a host with a virus can infect its direct healthy neighbors with rate $\beta$, while an infected host can be cured at rate $\delta$, after which the node becomes healthy, but susceptible again to the virus. The probability $v_{i}(t)$ depends on the probabilities $v_{j}(t)$ of the neighbors $j$ of node $i$ and there is no trivial closed form expression for $v_{i}(t)$. The model incorporates the network topology and is thus more realistic than the related population dynamic models. The model relies on the network topology, which makes it more realistic than the related population dynamic models. The goodness of the model has been evaluated in~\cite{PhysRevE.91.032812}. The probability of a node being infected in the metastable regime, denoted by $v_{i \infty}$, where $\frac{dv_{i}\left( t\right) }{dt} = 0$ and $v_{i \infty} \neq 0$, follows from (\ref{N_intertwined}) as~\cite{Omic09},
\small{
\begin{align}
v_{i \infty} &= 1 - \frac{1}{1+ \tau \sum_{j=1}^N a_{ij} v_{j \infty}} \label{expression_steadyState}
\end{align} 
}\normalsize where $\tau = \frac{\beta}{\delta}$ is called the \emph{effective infection} rate. The epidemic threshold $\tau_c$ is defined as a value of $\tau$, such that $v_{i \infty} >0$ if $\tau > \tau_c$, and otherwise $v_{i \infty} =0$  for all $i \in \{1,2,\ldots, N\}$. The value of  $v_{i \infty}$ depends of the values of all $v_{j \infty}$ for all the neighbors $j$ of $i$, so the network topology and the interconnectivity have impacts on $v_{i \infty}$s.

\subsection{Game-formation model}
In our network formation game, each player $i$ (a node in the network) aims to minimize its own \emph{cost function} $J_i$, and the \emph{social cost} $J$ is defined as $J = \sum_{i=1}^N J_i$. Specifically, the \emph{optimal social cost} is the smallest social cost over all possible connected topologies. We look for the existence, uniqueness, and characterization of (\emph{pure}) \emph{Nash Equilibria}\footnote{A Nash Equilibrium is the state of the players' network strategies, where none of the players can reduce its cost by unilaterally changing its strategy.}. The \emph{Price of Anarchy} (\emph{PoA}) and the \emph{Price of Stability} (\emph{PoS}) are defined as the ratio of social cost in the worst-case Nash Equilibrium (the one with highest social cost) and the optimal social cost, and the ratio of the social cost in the best-case Nash Equilibrium (the one with lowest social cost) and the optimal social cost, respectively:
\small{
\begin{align}
\text{PoA} = \frac{J(\text{worst NE})}{\min J}, \text{ PoS} = \frac{J(\text{best NE})}{\min J}. \label{PoA_PoS_definitions}
\end{align}
}\normalsize PoA is an efficiency measure, illustrating how bad selfish playing is, in comparison to the global optimum. PoS, on the other hand, reflects the best possible performance without coordination in comparison to the global optimum. The network about to be designed, is empty and every node in the network is a player. We assume the cost of building one (communication) link between two nodes is fixed. Every player $i$ can install a link from itself to another node $j$. Installing a link between $i$ and $j$ means that both $i$ and $j$ can utilize it, but only one pays for the cost, like often assumed in NFG models~\cite{Fabrikant:2003:NCG:872035.872088,Albers:2006:NEN:1109557.1109568,Chun2004}. Several examples fit this scenario: (i) a friend request is initiated by one node in a social network, but both read the posts from one another; (ii) a new road connecting two cities is built by one city in a road network, but both utilize it; and (iii) in a hand-shake protocol in a computer network one node initiates a connection used by two nodes.

We consider a \textbf{\emph{Virus Spread-Performance-Cost} (\textsc{VSPC})} network formation game, where player $i$ aims to reduce its cost and the probability $v_{i\infty}$ of being infected, but concurrently also wants to improve its performance by shortening the hopcounts $h(i,j)$ of the shortest paths to all the other nodes $j$. The cost function of player $i$ that unites these objectives is given by:
\small{
\begin{align}
J_i = \alpha \cdot k_i + \gamma \sum_{j=1}^{N} h(i,j) + v_{i\infty}. \label{node_costFunction}
\end{align}
}\normalsize Function $J_i$ involves the cost $k_i$ of installing all the links from node $i$, weighted by a coefficient $\alpha$. The hopcounts $h(i,j)$ are weighted by $\gamma$. Opposing goals meet in this game: the more links are installed, the shorter the paths, but the higher the probability of being infected and the higher the cost.

The social cost $J$ for the whole network is a weighted sum over all nodes
\small{
\begin{align}
J &
 = \sum_{i=1}^{N} J_i = \alpha L + \gamma\sum_{i=1}^N \sum_{j=1}^{N} h(i,j) +  \sum_{i=1}^N v_{i\infty},  \label{social_optimumPerformance}
\end{align}
}\normalsize 
where $L$ denotes the number of links.

\section{Virus spread-performance-cost (VSPC) game}\label{sec:networkFormationVirusSpreadPerformance}

\subsection{Optimal social cost, Nash Equilibria and the PoA for $\gamma \rightarrow 0$}

In order to understand the effect of the virus protection, we start by setting $\gamma$ to an infinitely small number (approaching zero\footnote{The case of $\gamma = 0$ is either trivial or debatable. By neglecting the hopcounts, the optimal topology would be the (non-realistic) empty graph with no links (cost) and no epidemic to be propagated. Moreover, infinite hopcounts will be multiplied by $\gamma = 0$ which is undefined.}). As a result, the hopcounts are of no influence anymore, while network connectivity is still guaranteed (the hopcount between two disconnected nodes is assumed to be infinity). 
%
%
Lemma~\ref{lemma:IncreasingProbabilities} limits the possible Nash Equilibria.

\begin{lemma}\label{lemma:IncreasingProbabilities} The probability $v_{i \infty} (G)$ of node $i$ being infected in the metastable state in network $G$ does not exceed the probability $v_{i \infty} (G+l)$ of node $i$ being infected in the metastable state in network $G+l$ obtained by adding a link $l$ to $G$.
\end{lemma}
\begin{proof}
The newly added link $l=(a,b)$ is between nodes $a$ and $b$. We make use of the canonical infinite form~\cite{Omic09},
\small{
\begin{align}
& v_{i \infty} = 1 - \frac{1}{1+\tau d_i - \tau \sum_{j=1}^N \frac{a_{ij}}{1+ \tau d_{j} - \tau \sum_{k=1}^N \frac{a_{jk}}{1+ \tau d_k - \ddots}}}. \label{canonicalExpression}
\end{align}
}\normalsize After the addition of link $l = (a,b)$, the expression (\ref{canonicalExpression}) for $v_{i \infty} (G+l)$ has all the terms the same as in $v_{i \infty} (G)$, except the following differences: $d_a \rightarrow d_a + 1$; $d_b \rightarrow d_b + 1$ and the presence of the adjacency entry $a_{ab} = 0 \rightarrow a_{ab} = 1$ in the canonical representation. The last statement implies that its contribution is a part that is the same as in $v_{i \infty} (G)$ until it ``reaches" nodes $a$ or $b$, where the expression (at a certain depth of the canonical form) is:
\footnotesize{
\begin{align}
&\tau (d_a + 1) - \frac{\tau}{\tau (d_b+1) - \ldots} + U
= \tau d_a + U  + \tau(1-\frac{1}{\tau (d_b+1) - \ldots}), \label{eq:DifferenceAddingLink}
\end{align}
}\normalsize where $d_a$ and $d_b$ are the degrees of $a$ and $b$ in $G$, while $U$ is the remaining part in the canonical form. In (\ref{eq:DifferenceAddingLink}), the term $\tau(1-\frac{1}{\tau (d_b+1) - \ldots})$ is positive and $U$ increases with $d_a$ and $d_b$. $U$ increases with $d_a$ and $d_b$ as it is also an infinite canonical form with any of these two variables being in the numerator or in the denominator with a negative sign in front, in the same way as explained in the lines above - repeating infinitely many times. Therefore, the whole term in (\ref{eq:DifferenceAddingLink}) increases, which implies that $v_{i \infty} (G+l) > v_{i \infty} (G)$ also increases for each node $i$.
%
\end{proof}

We start by looking into the possible Nash Equilibria.
\begin{theorem}\label{theorem:NashEqulibriumMustBeTrees}
If a Nash Equilibrium is reached, then the constructed graph is a \emph{tree}.
\end{theorem}
\begin{proof}
If $G$ is connected and each node can reach every other node, then changing the strategy of node $i$ from the current one to \emph{investing in an extra link}, will increase both its cost (by $1$, scaled by $\alpha$) and $v_{i\infty}$ (by Lemma~\ref{lemma:IncreasingProbabilities}). Hence, unilaterally \emph{investing in an extra link} is not beneficial for a node.

We now assume that $G$ is not a tree. Then, there is at least one cycle in this graph. If a node $i$ in that cycle changes its strategy from \emph{investing} in a link in that cycle to \emph{not investing}, the cost is decreased by $1$ (weighted by $\alpha$) and all the other nodes in the graph are still reachable from $i$. Moreover, by \emph{not investing} in that link, node $i$ decreases its probability $v_{i \infty}$ of being infected in the metastable state, according to Lemma~\ref{lemma:IncreasingProbabilities}. Hence, by unilaterally changing its strategy, node $i$ decreases its cost utility $J_i$, which is in contradiction with a Nash Equilibrium.
\end{proof}

\begin{observation}\label{observ:WhichTreesNE}
A Nash Equilibrium is achieved for both the \emph{star} graph and the \emph{path} graph, but not all trees are Nash Equilibria.
\end{observation}
\begin{proof}
Let us consider a \emph{star graph}, where all the links are installed by the root node as shown in Fig.~\ref{fig:Star}. (A link is installed and paid for by the node marked with $p$.) The root node cannot unilaterally decrease its cost, because cutting at least one of its installed links would disconnect it, while installing a link from a leaf node $i$ would increase both $k_i$ and $v_{i \infty}$ (Lemma~\ref{lemma:IncreasingProbabilities}). Hence, the \emph{star graph} is a Nash Equilibrium. 

\begin{figure*}[h!tb]
\centering
\subfloat[Star $K_{1,N-1}$.]{\hspace{-1cm}
\includegraphics[trim = 0mm 0mm 0mm 00mm,clip,width=0.2\textwidth]{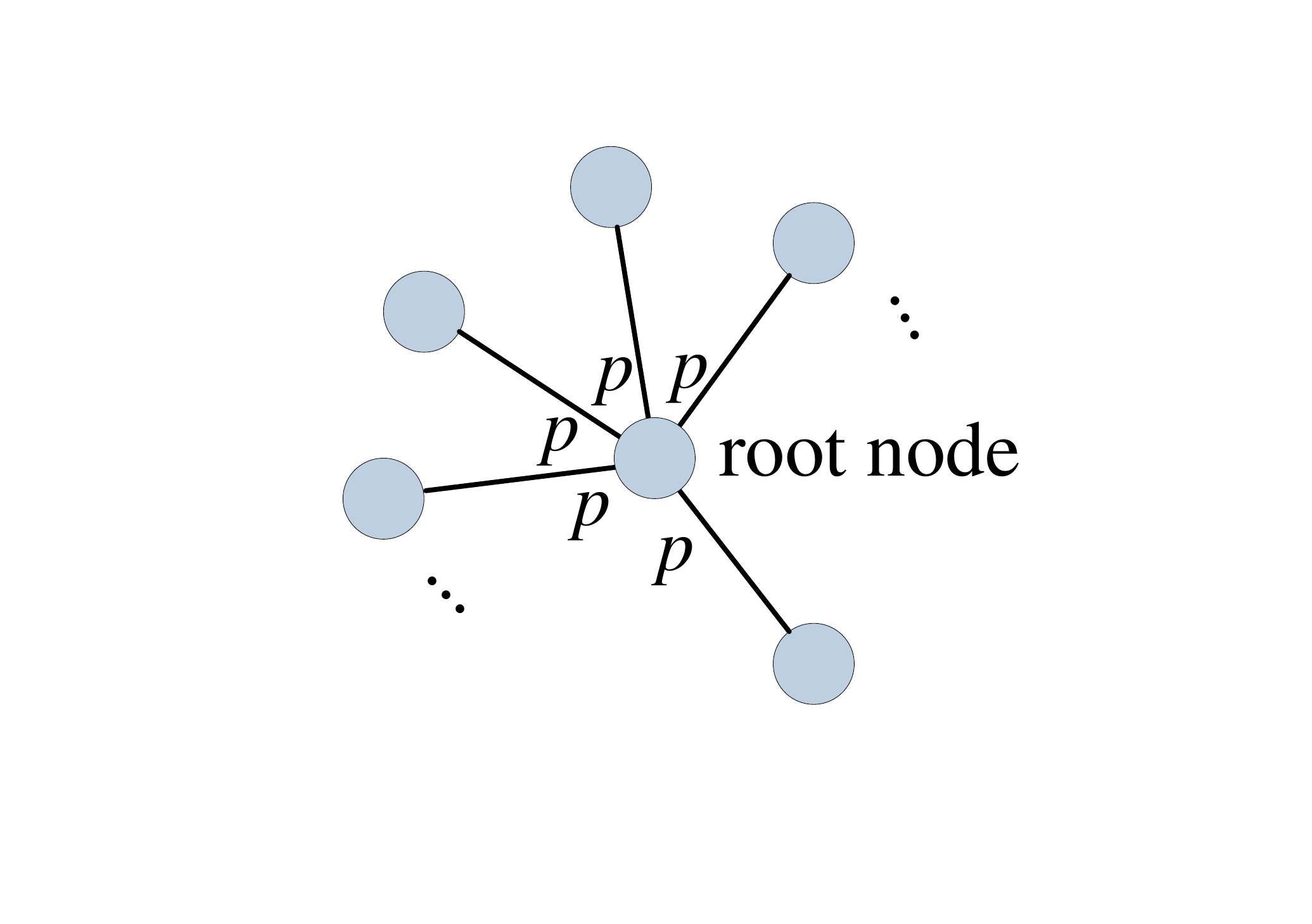}
\label{fig:Star}}
\subfloat[Path $P_{N}$.]{
\includegraphics[trim = 0mm 0mm 0mm 00mm,clip,width=0.2\textwidth]{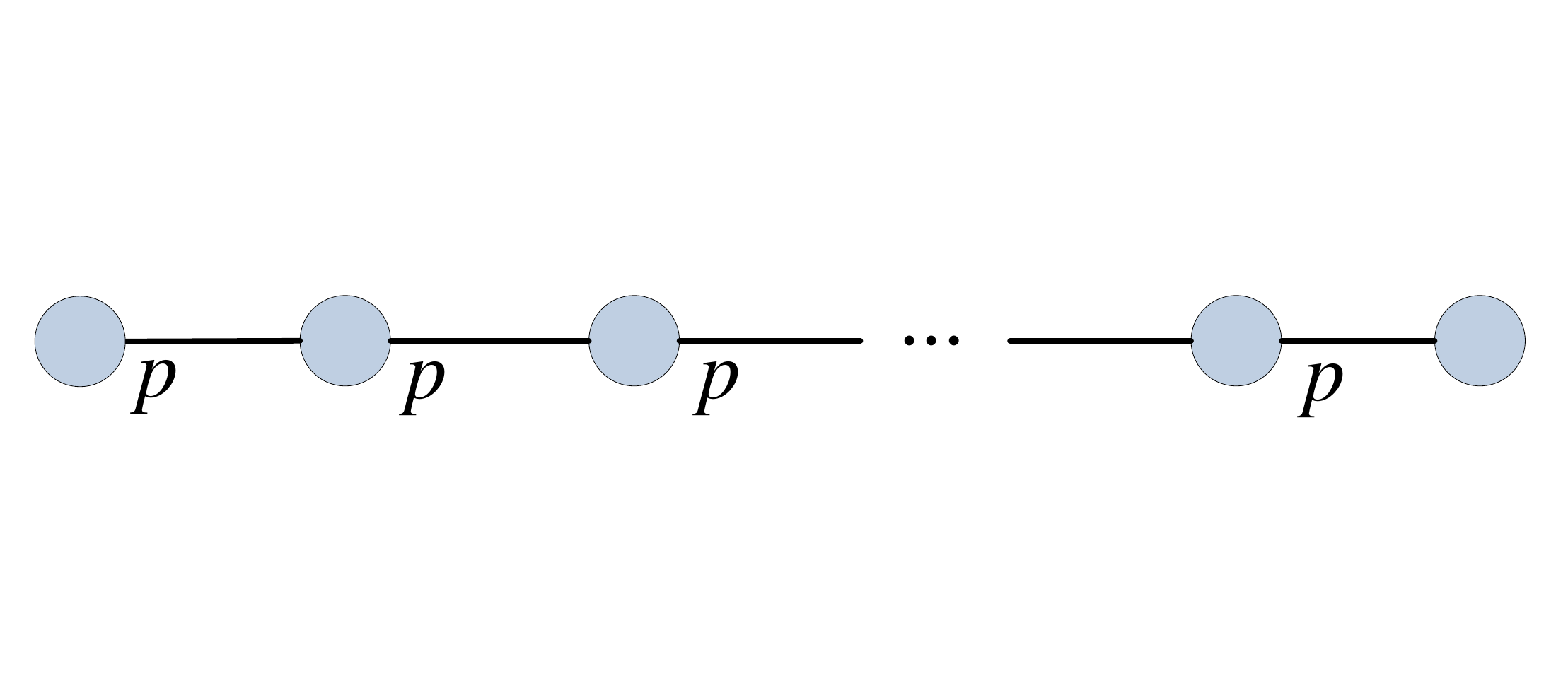}
\label{fig:Path}}
\subfloat[Tree $T'$.]{
\includegraphics[trim = 0mm 0mm 0mm 00mm,clip,width=0.2\textwidth]{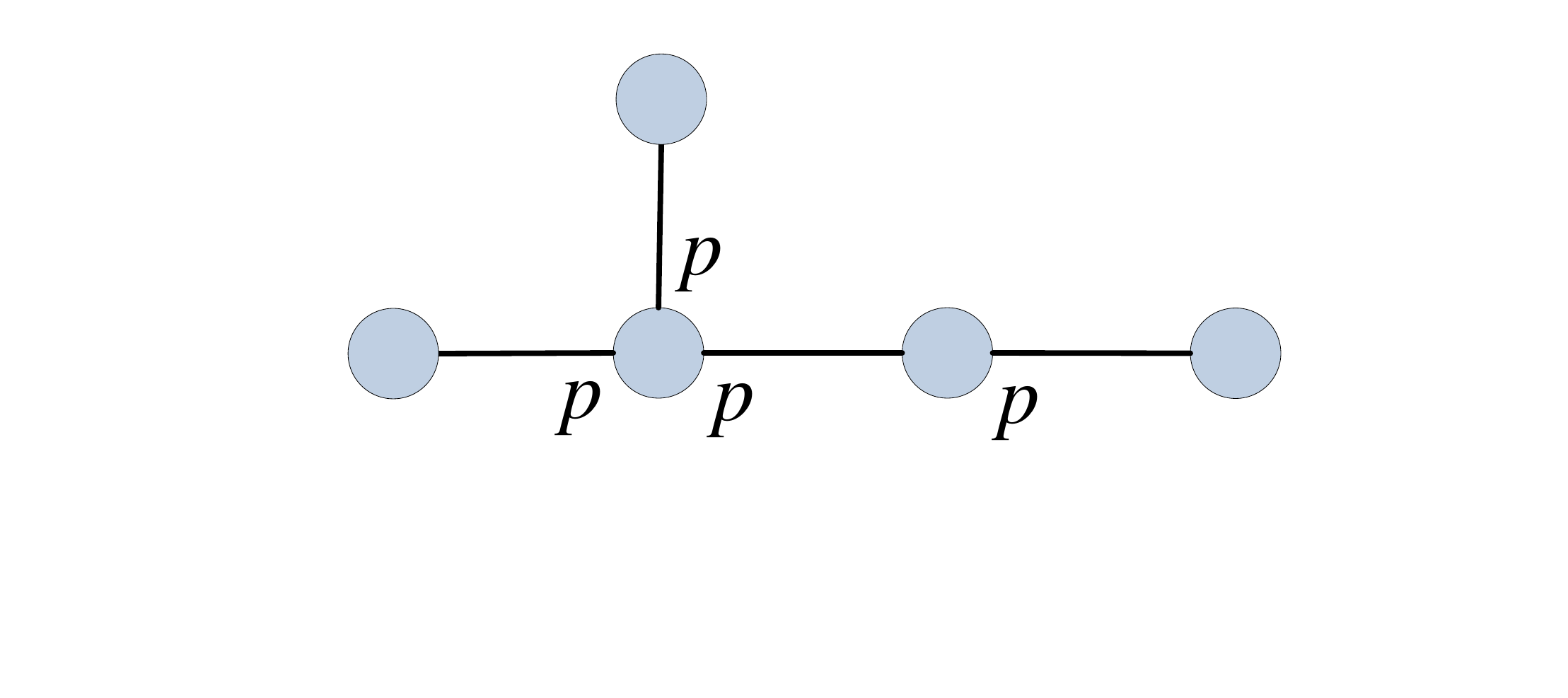}
\label{fig:tree5}}
\subfloat[Tree $T''$.]{
\includegraphics[trim = 0mm 0mm 0mm 00mm,clip,width=0.2\textwidth]{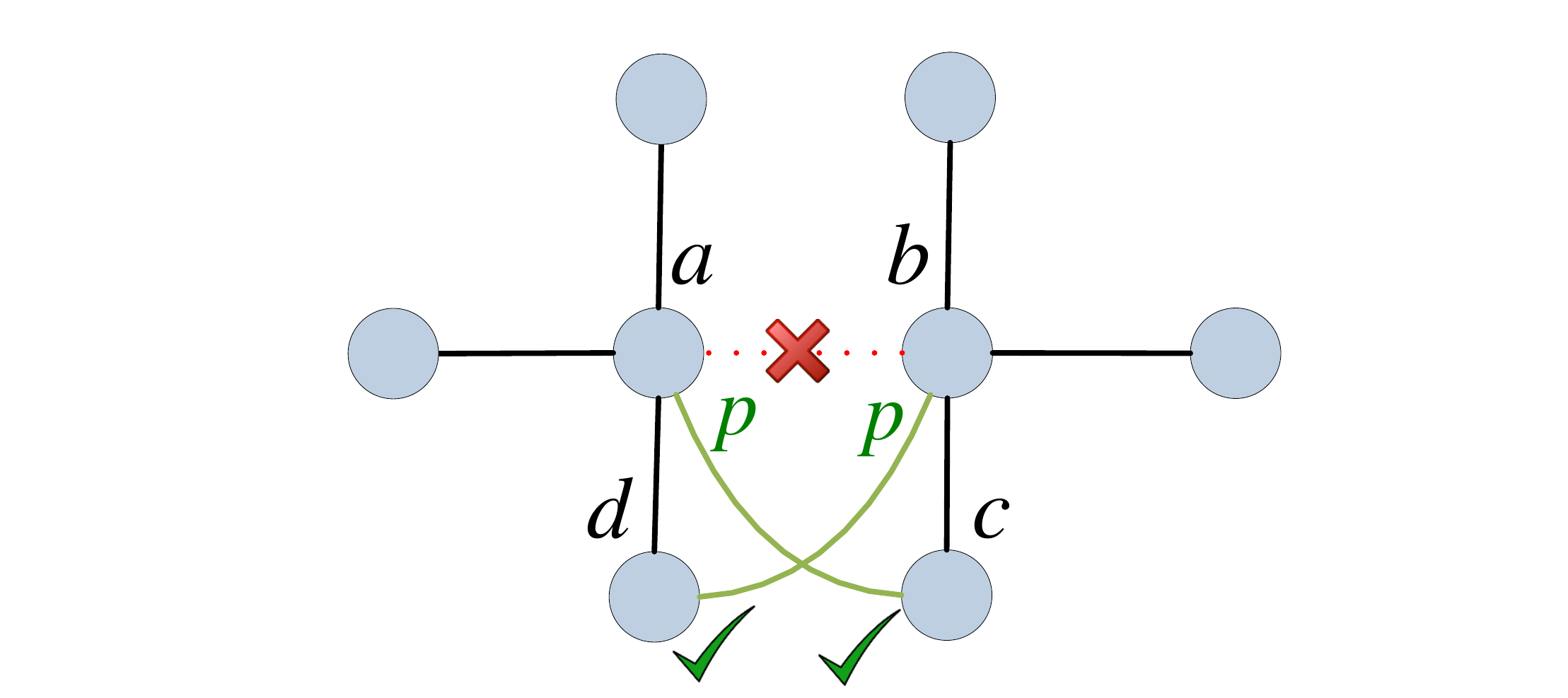}
\label{fig:treeNonExistenceNE}}
\subfloat[Re-wiring increases $v_{i \infty}$.]{
\includegraphics[trim = 0mm 0mm 0mm 00mm,clip,width=0.2\textwidth]{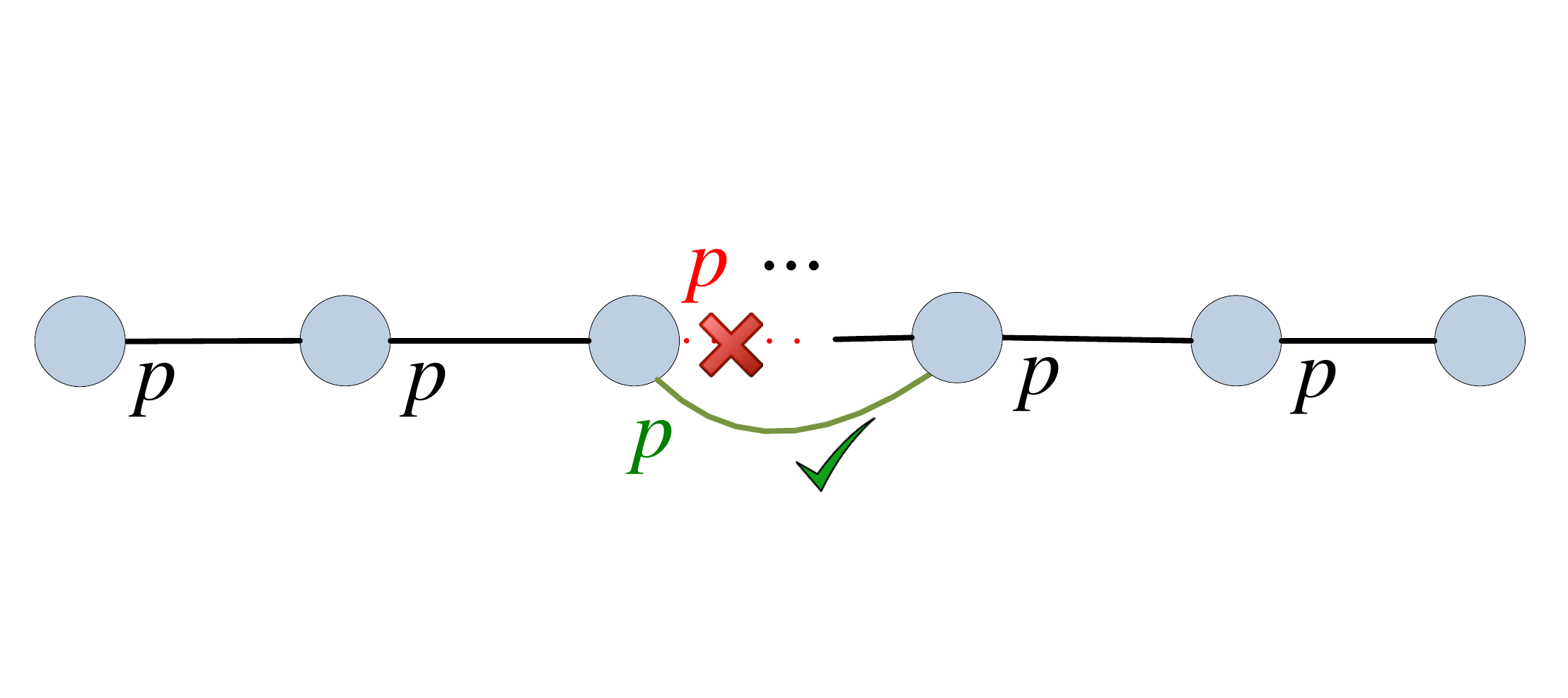}
\label{fig:treeRewiring}}
\caption[]{A link is installed by the end-node marked with $p$. Trees in (a), (b), and (c) are Nash Equilibria. (d) Tree $T''$ cannot be a Nash Equilibrium.} \label{fig:NashEquilibriaObservationsTrees}
\end{figure*}

Let us now assume that a \emph{path graph} (Fig.~\ref{fig:Path}) is constructed, such that $(N-1)$ nodes invest in exactly one link and one of the leaves does not invest in installing a link. Similarly as for a \emph{star graph}, none of the nodes can unilaterally decrease their cost by just installing extra links or cutting some of them. A "re-wiring''\footnote{"Re-wiring" is a process of removing a link to node $k$ initiated by node $i$ and establishing a new link to another node $j$. The degree of node $i$ does not change, while the degrees of $k$ and $j$ are decreased and increased, respectively.} from one of the nodes by re-directing its installed links to another node may be in order. In such a case, if node $i$ ``re-wires'' its installed link to another node, then $J_i$ would not decrease. 1) If it is installed to one of the leaves, such that the graph is connected, we end up with an isomorphic graph, where the position of $i$ is the same as in the initial graph, so $J_i$ stays the same. 2) If $i$ "re-wires'' to one of the other nodes $j$ (w.l.o.g., $i<j$) as visualized in Fig.~\ref{fig:treeRewiring}, $i$ would have the same degree, but its ``new neighbor'' would have a degree $3$ instead of $2$. The degree of $j$ increases by $1$ to $3$ and the degree of $(i+1)$ decreases by $1$ to $1$ (node $(i+1)$  will become terminal and "far" from $i$), while all the other degrees remain the same. Moreover, $i$ would be equally close to any of the nodes ``behind'' $\{1,\ldots,i-1\}$, closer to the nodes ``at the end'' $\{j+1,\ldots,N\}$ and equally close to the nodes in the set $\{i+1,\ldots,j-1\}$, but just in a reverse order. Based on the canonical infinite form (\ref{canonicalExpression}), $v_{i \infty}$ would increase\footnote{$v_{i \infty}$ in (\ref{canonicalExpression}) would have bigger values by having nodes with "bigger degrees" as as close as possible (i.e. in fewer hops) to the node.}. Therefore, the \emph{path graph} is also a Nash Equilibrium. 

There are also other \emph{trees} that are Nash Equilibria (e.g., $T'$ given in Fig.~\ref{fig:tree5}). Moreover, there are values of $\tau$ such that worst- and best-case Nash Equilibria are achieved for trees different from star $K_{1,N-1}$ and path $P_N$ graphs. For $\tau \in [1.475, 1.589]$, tree $T'$ is the best-case Nash Equilibrium and has optimal social cost.

However, not all the trees are Nash Equilibria. For example, the tree given in Fig.~\ref{fig:treeNonExistenceNE}. Here, whomever pays for the ``central'' link between $a$ and $b$, can reduce its cost utility by ``re-wiring'' to $c$ or $d$.
\end{proof}
We proceed by characterizing the worst- and best-case Nash Equilibria.

\begin{theorem}\label{upperAndLowerBound_J_OnTrees}
For sufficiently high effective infection rate $\tau$, the optimal social cost and the \emph{best-case} Nash Equilibrium are achieved by the \emph{star graph} $K_{1,N-1}$, while the \emph{worst-case} Nash Equilibrium is achieved for the \emph{path graph} $P_N$,
\begin{align*}
J(K_{1,N-1}) \le J \le J(P_N).
\end{align*}
\end{theorem}
\begin{proof}
According to Theorem~\ref{theorem:NashEqulibriumMustBeTrees}, in a Nash Equilibrium the graph is a tree, hence it has $N-1$ links. In a general case, from a tree in which there are two nodes $i$ and $k$, connected to one another, for which $d_i \ge 3$ and $ d_k =1$ (i.e. $k$ is a leaf), by breaking the connection between $i$ and $k$ and connecting $k$ to another leaf $j$ instead, we have: the degree of $k$ is $1$ (remains the same);  the degree of node $i$ becomes $d_i-1 \ge 2$ (decreased by one); and the degree of $j$ is $2$ (increased by one). The process can be repeated until there exists a node of degree at least $3$ in the tree. At the end, we end up with a tree with no degree bigger than $2$ and this is a path $P_N$. The social cost $J$ is increased in each step~\cite[Lemma 2]{CDC2015_GameFormationVirusSpread}. In this way, the process converges to a path $P_N$.

In a very similar (but reverse) process, starting from any tree $G$, we can decrease $J$ at each step, ending up with a star $K_{1,N-1}$ with a maximum $J(G)$ in the final step.
\end{proof}

However, what would be the optimal social cost, and the worst- and best-case Nash Equilibria highly depends on the effective infection rate $\tau$.

\begin{theorem}\label{upperAndLowerBound_J_OnTreesSmallTau}
For low values of the effective infection rate $\tau$, above but sufficiently close to the epidemic threshold $\tau_c$, the optimal social cost and the \emph{best-case} Nash Equilibrium are achieved by the \emph{path graph} $P_N$, while the \emph{worst-case} Nash Equilibrium is achieved by the \emph{star graph} $K_{1,N-1}$,
\begin{align*}
 J(P_N) \le J \le J(K_{1,N-1}). 
\end{align*}
\end{theorem}
\begin{proof}
We consider a spectral approach~\cite{PAVanMieghem2} and denote $y(\tau) = \sum \limits_{i=1}^N v_{i\infty} (\tau)$ the infection probability of all nodes in the metastable state. The probabilities of a node in the graph being infected are non-zero and $y(\tau) > 0$  if $\tau >\tau_c = \frac{1}{\lambda_1}$, where $\lambda_1$ is the largest eigenvalue of the adjacency matrix in the graph~\cite{Ganesh2005}. For $ \tau < \frac{1}{\lambda_1}$, $y(\tau) = 0$.

Lov\'{a}sz and Pelik\'{a}n~\cite{LovaszPelikan1973} ordered all the trees with $N$ nodes by the largest eigenvalues of the adjacency matrices. It turns out that, the path $P_N$ and star $K_{1,N-1}$ are the trees with the minimum $\lambda_1 (P_N)$ and maximum $\lambda_1 (K_{1,N-1})$ largest eigenvalues, respectively.

For values $\tau = \frac{1}{\lambda_1 (K_{1,N-1})} + \varepsilon = \frac{1}{\sqrt{N-1}} + \varepsilon$, it holds that $y_{K_{1,N-1}}(\tau) > y_{T}(\tau) = 0$, where $T$ is any tree different from $K_{1,N-1}$, therefore $J(K_{1,N-1}) $ is the largest.

For values $\tau = \frac{1}{\lambda_1 (P_N)} - \varepsilon = \frac{1}{2\cos(\frac{\pi}{N+1})} - \varepsilon$, we have $y_{T}(\tau) > y_{P_N} (\tau) = 0$, where $T$ is any tree different from $P_{N}$, hence $J(P_N) $ is the smallest.
\end{proof}

Theorems~\ref{upperAndLowerBound_J_OnTrees} and \ref{upperAndLowerBound_J_OnTreesSmallTau} show opposite behavior depending on whether the value $\tau$ is in the high or low regime, although both revolve around the path and star graphs. For $\tau$ in the intermediate regime, different trees may give the best-/worst-case Nash Equilibrium.


\begin{corollary}\label{corr:PoSPoA}
For both high and low effective infection rate $\tau$, $\text{PoS} = 1 \text{ and }\\ \text{PoA} = \max \big\{\frac{J(P_N)}{J(K_{1,N-1})},\frac{J(K_{1,N-1})}{J(P_N)} \big \}$.
\end{corollary}
\begin{proof}
Based on Theorems~\ref{upperAndLowerBound_J_OnTrees} and \ref{upperAndLowerBound_J_OnTreesSmallTau}, for high (low) $\tau$, tree $K_{1,N-1}$ ($P_N$) is both optimal in social cost and the best-case Nash Equilibrium, while $P_N$ ($K_{1,N-1}$) is the worst-case Nash Equilibrium. Based on the definitions for PoS and PoA in (\ref{PoA_PoS_definitions}), $\text{PoS} = \frac{J(K_{1,N-1})}{J(K_{1,N-1})} = 1 (= \frac{J(P_N)}{J(P_N)})$; and PoA$=\frac{J(P_N)}{J(K_{1,N-1})}$ for large enough $\tau$ and PoA$=\frac{J(K_{1,N-1})}{J(P_N)}$ for $\tau$ close to the epidemic threshold $\tau_c$.
\end{proof}

\begin{corollary}\label{corr:PoSPoA2}
For sufficiently high effective infection rate $\tau$, in the virus spread-cost game formation,
\small{
\begin{align*}
\text{PoA} < 1+ \frac{1}{2\big (\tau(\alpha+1)-1 \big)}, \text{ where $\tau(\alpha+1)>1$.}
\end{align*}
}
\end{corollary}
\normalsize
\begin{proof}  The proof is provided in~\cite{CDC2015_GameFormationVirusSpread}.
\end{proof}

The exact value of the PoA is given in Fig.~\ref{fig:VSCN10} by making use of Corollary~\ref{corr:PoSPoA}. It is highest ($\sim3.3$) for small $\tau$, above the epidemic threshold and it further sharply decreases reaching $1$ for a unique Nash Equilibrium. For higher $\tau$, the PoA increases towards its maximum around $1.1$ and then it slowly decreases approaching $1$. 
\begin{figure}[h!tb]
\centering
\subfloat[$N=10$.]{\hspace{-0.35cm}
\includegraphics[trim = 12.8mm 68mm 22.8mm 68mm,clip,width=0.34\textwidth]{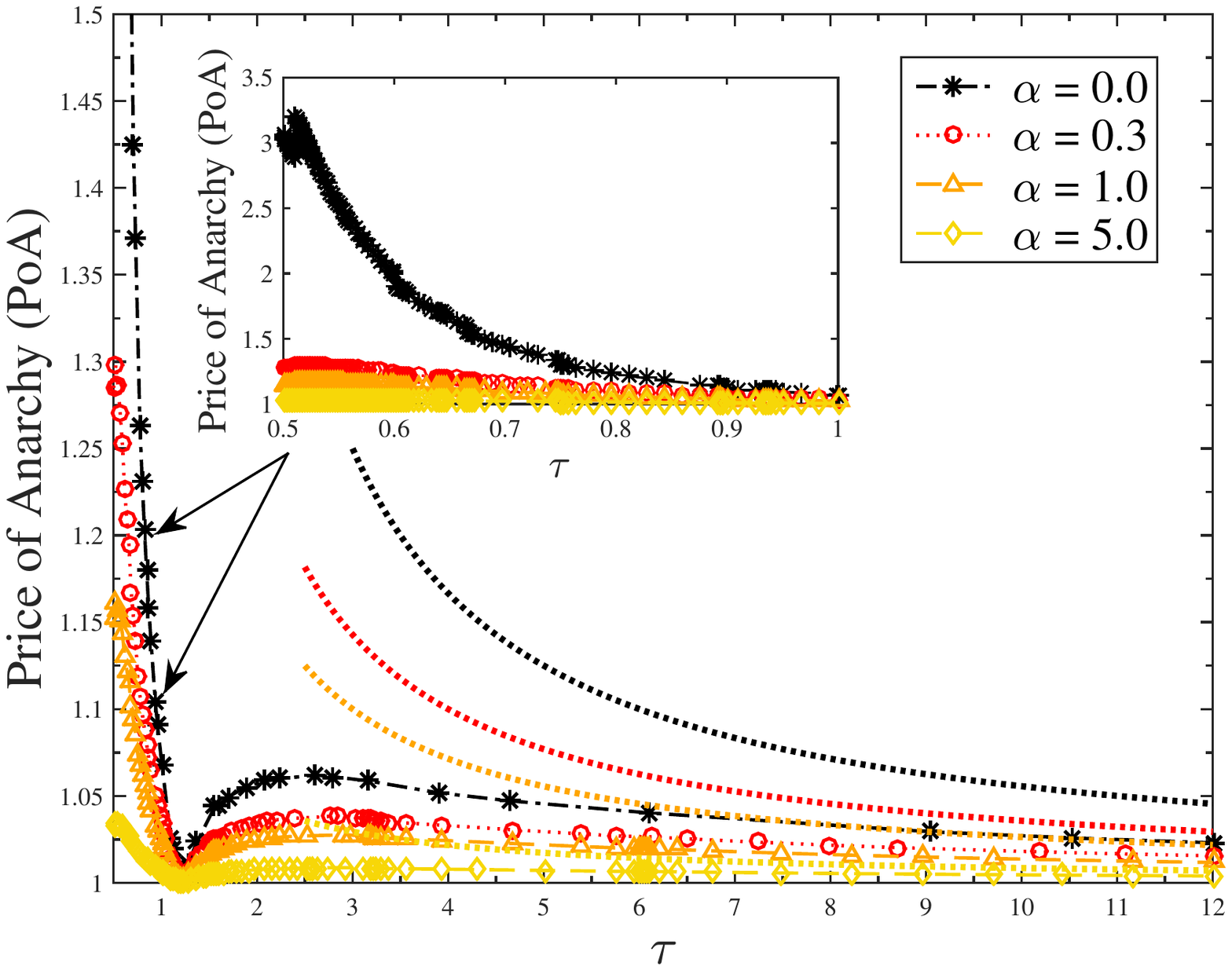}
\label{VSCN10}}\\
\subfloat[$N=1000$.]{
\includegraphics[trim = 12.8mm 68mm 22.8mm 68mm,clip,width=0.34\textwidth]{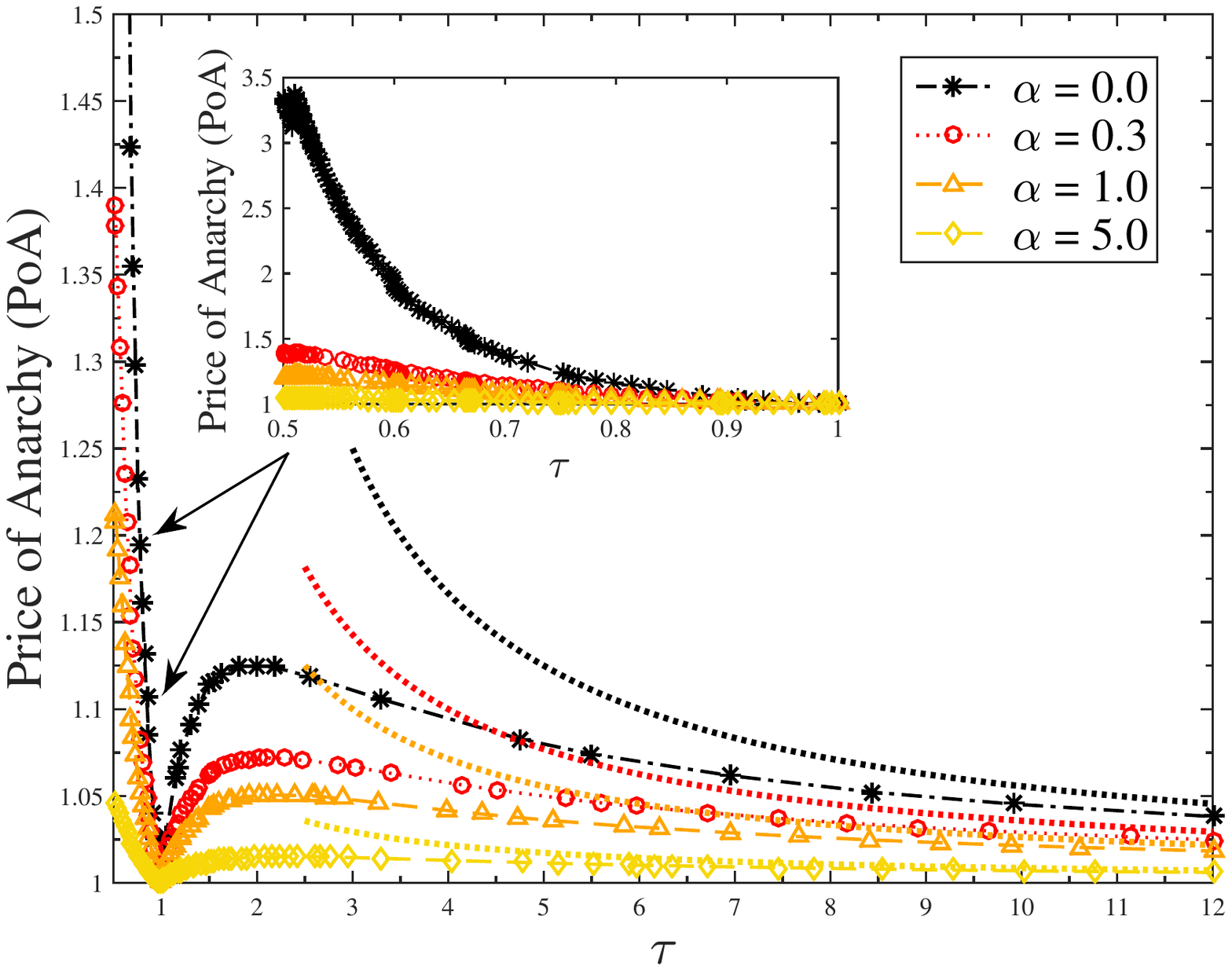}
\label{VSCN1000}}
\caption[]{The Price of Anarchy (PoA). The dotted lines represent the bound from Corollary~\ref{corr:PoSPoA2}.} \label{fig:VSCN10}
\end{figure}

%
%
We have observed that the equilibria tree topology in which a virus thrives is not always a star (i.e., the tree with the smallest diameter), but that it may differ with the virus infection rate. For most of the $\tau$ values (except maybe small $\tau$), a small value for the Price of Anarchy (PoA) means that a topology close to optimal can be obtained in a decentralized manner, even when the individual players play selfishly.

\subsection{Optimal social cost, Nash Equilibria and the PoA for $\gamma > 0$}

We start by analyzing the social cost (\ref{social_optimumPerformance}). 
 Node $i$ is one hop away from its $d_i$ neighboring nodes, while it is at least $2$ hops away from the other $N-1-d_i$ nodes, hence $\sum_{j=1}^{N} h(i,j) \ge d_i + 2(N-1-d_i)$. Using this, for large enough $\tau$ when $\sum_{i=1}^N v_{i\infty}$ can be approximated\footnote{In fact, the sum can be lower bounded~\cite[p. 10]{Omic09} by $\sum_{i=1}^N v_{i\infty} > N - \sum_{i=1}^N \frac{1}{1+(\tau - 1) d_i}$, which is meaningful for $\tau>1$.} by using truncation of Maclaurin seria~\cite[Lemma 1]{VanMieghem2011}, the social cost in (\ref{social_optimumPerformance}) is lower bounded as
\small{
\begin{align}
J & \ge N+ 2\gamma N (N-1)+ (\alpha-2 \gamma) L  -\frac{1}{\tau} \sum_{i=1}^N \frac{1}{d_i}. \label{socialCostWiener}
\end{align}
}\normalsize The following bound is due to Cioab\u{a}~\cite[Theorem 9]{Cioaba20061959},
\small{
\begin{align}
\sum_{i=1}^N \frac{1}{d_i} &\le \frac{N^2}{2L} + (\frac{1}{d_{\min}}- \frac{1}{d_{\max}})(N-1-\frac{2L}{N}), \notag
\end{align}
}\normalsize where the equality holds for regular graphs and the star graph. Based on this, $d_{\min} \ge1$, and $d_{\max} \le N-1$, we obtain
\small{
\begin{align}
\sum_{i=1}^N \frac{1}{d_i} & \le \frac{N^2}{2L} + (1- \frac{1}{N-1})(N-1-\frac{2L}{N}) \notag \\
&= \frac{N^2}{2L} + \frac{N-2}{N(N-1)}(N(N-1)-2L). \label{boundInverse degree}
\end{align}
}\normalsize Equality in (\ref{boundInverse degree}) is achieved only for the star $K_{1,N-1}$, where $d_{\max} = N-1$ and $d_{\min} = 1$, or for the complete graph $K_{N}$ (where $2L = N(N-1)$). (The equality for other regular graphs is ruled out because of the inequality in (\ref{boundInverse degree}).) Using (\ref{boundInverse degree}) into (\ref{socialCostWiener}) yields
\footnotesize{
\begin{align}
&J \ge  N+ 2\gamma N (N-1) - \frac{N-2}{\tau} +(\alpha-2 \gamma + \frac{2(N-2)}{\tau N (N-1)}) L -\frac{N^2}{2\tau L} \label{socialCostWiener1}
\end{align}
}\normalsize Let us consider two regimes:

\begin{enumerate}
\item If $\alpha-2 \gamma + \frac{2(N-2)}{\tau N (N-1)} \ge 0$, then the bound in (\ref{socialCostWiener1}) is an increasing function in $L$, hence the optimal social cost is achieved for $L = N-1$. The bound in (\ref{socialCostWiener1}) is tight for such $L$, because the bounds in (\ref{boundInverse degree}) and (\ref{socialCostWiener}) become equalities for $K_{1,N-1}$ and any graph with a diameter at most two, respectively. Hence, $J \ge J (K_{1,N-1})$ and equality is achieved only for the star graph $K_{1,N-1}$.

\item If $\alpha-2 \gamma + \frac{2(N-2)}{\tau N (N-1)} < 0$, then the bound in (\ref{socialCostWiener1}) increases for $L< \frac{N}{\sqrt{2\tau(2\gamma - \alpha) - \frac{4(N-2)}{N (N-1)})}}$ and decreases for $L> \frac{N}{\sqrt{2\tau(2\gamma - \alpha) - \frac{4(N-2)}{N (N-1)})}}$. Hence, the optimal social cost is achieved in one of two boundary cases: $L = N-1$ and $L = \binom{N}{2}$. For $L = N-1$, similarly as in 1), we obtain that the only possibility is the star graph $K_{1,N-1}$, while for $L = \binom{N}{2}$ it is the complete graph $K_{N}$. Finally, $J \ge \min \{J(K_{1,N-1}), J (K_N)\}$.

It remains to compare $J(K_{1,N-1})$ and $J(K_N)$: $J(K_{1,N-1}) = N+ \alpha(N-1) + 2 \gamma (N-1)^2 - \frac{(N-1)^2+1}{\tau(N-1)}$ and $J(K_{N}) = N+ \alpha \frac{N(N-1)}{2} + \gamma N(N-1) - \frac{N}{\tau(N-1)}$. Hence,
\small{
\begin{align*}
J(K_{N}) - J(K_{1,N-1}) = (N-1)(N-2) (\frac{\alpha}{2} - \gamma + \frac{1}{\tau (N-1)}).
\end{align*}
}\normalsize If $\alpha \le 2 \gamma - \frac{2}{\tau (N-1)}$, then $J(K_{N}) \le J(K_{1,N-1})$ and the optimal social cost is achieved for the complete graph $K_N$. If $\alpha \ge 2 \gamma - \frac{2}{\tau (N-1)}$, then $J(K_{N}) \ge J(K_{1,N-1})$ and the optimal social cost is achieved for the star graph $K_{1,N-1}$. The last also covers case 1), because $2 \gamma - \frac{2}{\tau (N-1)} < 2 \gamma - \frac{2(N-2)}{\tau N (N-1)}$.  
\end{enumerate}
Now, for the optimal social cost, Theorem~\ref{socialCostPerformranceGame} follows.

\begin{theorem}\label{socialCostPerformranceGame} For sufficiently high $\tau$, the optimal social cost is achieved for the star $K_{1,N-1}$ if $\alpha \ge 2 \gamma - \frac{2}{\tau (N-1)}$, and for the complete graph $K_N$, otherwise.
\end{theorem}

We proceed with characterization of the Nash Equilibria and the Price of Anarchy for sufficiently high $\tau$. In the \textsc{VSPC} game, Nash Equilibria topologies can be complex, while the star and the complete graph can appear as extreme cases:
\begin{itemize}
\item The complete graph $K_N$ is a Nash Equilibrium, if and only if $\alpha \le \gamma - \frac{1}{\tau (N-1)}$. Since new links cannot be added, changing the strategy for a node $i$ means deleting $k$ of its links ($1\le k \le N-2$). The corresponding change would increase the cost $J_i$ of $i$ by $k(\gamma - \alpha) - \frac{1}{\tau (N-1-k)} + \frac{1}{\tau (N-1)} = k(\gamma - \alpha)  - \frac{k}{\tau(N-1)(N-1-k)} \ge \frac{k}{\tau (N-1)} (1- \frac{1}{N-1-k}) \ge 0$. Hence, node $i$ has no interest to deviate from its current strategy. On the other hand,  if $\alpha > \gamma - \frac{1}{\tau (N-1)}$ and node $i$ changes its strategy by cutting $(N-2)$ links (all except one - to keep its connectivity), the change in $J_i$ is equal to $(N-2)(\gamma - \alpha - \frac{1}{\tau(N-1)}) <0$, which will reduce its cost.

\item The star graph $K_{1,N-1}$ is a Nash Equilibrium, if and only if $\alpha \ge \gamma - \frac{1}{\tau (N-1)}$. The root node cannot delete a link, because this would make its cost infinity. If $i$ is a leaf, for some $k \ge 0$, changing its strategy means: (i) adding $k$ links, then the hopcounts to these nodes are reduced from $2$ to $1$, hence the contribution from the hopcounts is changed by $-k\gamma$; or (ii) deleting the link installed by him (if any) and installing $(k+1)$ links, where $k+1<N-2$. In (ii), the hopcount to the root node is increased from $1$ to $2$, the hopcount to $(k+1)$ links is decreased from 2 to 1, and the hopcounts to the other $\big( (N-2) - (k+1)\big)$ nodes are increased from $2$ to $3$. The change in the sum of hopcounts is: $-(k+1) \gamma + 1\cdot\gamma + \big( (N-2) - (k+1)\big) \gamma = -k\gamma + \big( (N-2) - (k+1)\big)  \gamma \ge -k\gamma$, hence the change of the hopcount is again at least $-k\gamma$. Thus, the change in $J_i$ is at least $k \alpha - k \gamma - \frac{1}{\tau (k+1)} + \frac{1}{\tau} = k(\alpha -\gamma + \frac{1}{\tau (k+1)}) \ge k(\alpha -\gamma + \frac{1}{\tau (N-1)}) \ge 0$. On the other hand, if $\alpha < \gamma - \frac{1}{\tau (N-1)}$, the change in $J_i$ by adding $(N-2)$ links from one leaf to all the other leaves in $K_{1,N-1}$, is $(N-2)(\alpha- \gamma) - \frac{1}{\tau(N-1)} + \frac{1}{\tau} = (N-2)(\alpha - \gamma + \frac{1}{\tau (N-1)}) <0$, i.e. it is not a Nash Equilibrium.
\end{itemize}

The above two points resolves the conditions for two specific graphs, but they do not cover all the possibilities for the Nash Equilibria and the Price of Anarchy, which may vary on different intervals and a case analysis, as provided in the following, is required. We will consider the case $\alpha < 2\gamma - \frac{1}{\tau}$ and the case $\alpha > 2\gamma - \frac{1}{\tau}$.
\vspace{-1em}
\subsection*{Case $\alpha < 2\gamma - \frac{1}{\tau}$.} Now, $\gamma > \frac{1}{2\tau}$. A Nash Equilibrium is achieved only for graphs with a diameter at most $2$ - an argument used in the later points (b) and (c). The proof is by contradiction. Let us assume node $i$ is at least $3$ hops away from another node. Clearly, $d_i \le (N-2)$ and if $i$ installs a link from $i$ to $j$, the difference in $J_i$ is at least $\alpha - 2 \gamma +\frac{1}{\tau d_i (d_i +1)} \le - \frac{1}{\tau} + \frac{1}{\tau d_i (d_i +1)}  <0$. Hence, $i$ reduces its cost and the graph is not a Nash Equilibrium. We consider three sub-intervals (a), (b) and (c):

\hspace{-1em}\textbf{(a)} If $\alpha < \gamma - \frac{1}{2\tau}$, adding a link from $i$ will change $J_i$ by at least $\alpha - \gamma + \frac{1}{\tau d_i (d_i +1)} < - \frac{1}{2\tau} + \frac{1}{\tau d_i (d_i +1)} \le 0$. Therefore, the complete graph $K_N$ is the only Nash Equilibrium. Because, $\alpha < \gamma - \frac{1}{2\tau} \le 2\gamma - \frac{2}{\tau (N-1)}$ for $N \ge 3$ and, according to Theorem~\ref{socialCostPerformranceGame}, it also has optimal social cost. Finally, $\text{PoA} = \text{PoS} = 1$.

\hspace{-1em}\textbf{(b)} If $\gamma - \frac{1}{2\tau} \le \alpha \le \gamma - \frac{1}{\tau (N-1)}$ and we assume, by contradiction, that there is a Nash Equilibrium different from $K_N$, we have the following:

\begin{itemize}[noitemsep,nolistsep]
\item If there is a link in the graph, installed by node $i$ such that its deletion increases the sum of hopcounts from $i$ by only $1$, then $J_i$ is increased by: $\gamma - \alpha - \frac{1}{\tau d_i (d_i-1)}>0$. On the other hand, adding a link would change $J_i$ to: $ \alpha -\gamma + \frac{1}{\tau d_i (d_i+1)}>0$. The last two inequalities imply, $0 < \alpha -\gamma + \frac{1}{\tau d_i (d_i+1)} <- \frac{1}{\tau d_i (d_i-1)} + \frac{1}{\tau d_i (d_i+1)} = -\frac{2}{\tau (d_i-1) d_i (d_i+1)} < 0$, which is a contradiction. Hence, there is no other Nash Equilibrium different from $K_N$ and $\text{PoA} = \text{PoS} = 1$.

\item If deleting any of the links installed by $i$ would increase the sum of hopcounts by at least $2$; by link deletion, the difference in $J_i$ is at least $2 \gamma - \alpha - \frac{1}{\tau d_i (d_i-1)} $ and we have $2 \gamma - \alpha - \frac{1}{\tau d_i (d_i-1)} \ge \gamma + \frac{1}{\tau (N-1)} - \frac{1}{\tau d_i (d_i-1)} \ge \frac{1}{2\tau} -  \frac{1}{\tau d_i (d_i-1)} + \frac{1}{\tau (N-1)} \ge \frac{1}{\tau (N-1)} >0$. We proceed by considering the properties of the possible Nash Equilibria in particular sub-intervals:
$- \frac{1}{\tau (k-1) k} \le \alpha -\gamma < - \frac{1}{\tau k (k+1)}$ for $k \in \{2,3,\ldots,\lfloor \sqrt{N-\frac{3}{4}} - \frac{1}{2} \rfloor\} $. By link addition, the difference in $J_i$ is $\alpha- \gamma +\frac{1}{\tau d_i (d_i +1)}$ and a necessary condition for a Nash Equilibrium is $d_i < k$. On the other hand, $ k \le  \sqrt{N-\frac{3}{4}} - \frac{1}{2} \le \sqrt{N-1}$, hence $d_i < \sqrt{N-1}$. Therefore, we have less than $\sqrt{N-1}$ nodes that are on a distance $1$ from a node $i$. Each of these nodes is directly connected by less than $\sqrt{N-1}-1$ nodes different from $i$. Hence, there less than $\sqrt{N-1}+\sqrt{N-1}(\sqrt{N-1}-1) = N-1$ nodes that are at most $2$ hops from $i$, hence at least one node that is more than $2$ hops away from $i$, a contradiction to the general claim (before (a))! Hence, $K_N$ is the only Nash Equilibrium and $\text{PoA} = \text{PoS} = 1$.
\end{itemize}

\hspace{-1em} \textbf{(c)} If $\gamma - \frac{1}{\tau (N-1)} \le \alpha < 2\gamma - \frac{1}{\tau}$, then $K_{1,N-1}$ is a Nash Equilibrium. Graphs that are of diameter at most $2$ are also candidates for a Nash Equilibrium. 

Because the diameter of the graph is not bigger than $2$, (\ref{socialCostWiener}) becomes an equality $J = N+ 2\gamma N (N-1)+ (\alpha-2 \gamma) L  -\frac{1}{\tau} \sum_{i=1}^N \frac{1}{d_i}$ for sufficiently large $\tau$. Applying the condition of (c) leads to
\vspace{-1em}
\small{
\begin{align}
J(\text{worst NE}) & <  N+ 2\gamma N (N-1) - \frac{1}{\tau}  L  -\frac{1}{\tau} \sum_{i=1}^N \frac{1}{d_i} \notag\\
&= N+ 2\gamma N (N-1) - \frac{1}{\tau} \sum_{i=1}^N (\frac{d_i}{2} +\frac{1}{d_i}) \notag \\
& \le N+ 2\gamma N (N-1) - \frac{3N}{2\tau} = N(1+2\gamma (N-1) - \frac{3}{2\tau}) \label{socialCostUpperBound}
\end{align}
}\normalsize due to the fact that $\frac{d_i}{2} +\frac{1}{d_i} \ge \frac{3}{2}$ (equivalent to $(d_{i} -1)(d_{i} -2) \ge 0$). Equality holds (only) in the last line of (\ref{socialCostUpperBound}) if $d_i=1$ or $d_i=2$ for all $i$ (e.g., the ring $C_N$ or the path $P_N$ graphs), otherwise a strict inequality in the second part also holds. 
Finally, knowing that the optimal social cost is attained by the complete graph $K_N$ and the condition inequality condition in (c) for $J(K_N)$: $\text{PoA} = \frac{J(\text{worst NE})}{J(K_N)} < \frac{1+2\gamma (N-1) - \frac{3}{2\tau}}{1+ \frac{3\gamma (N-1)}{2} - \frac{1}{2\tau} -\frac{1}{\tau(N-1)}} = \frac{\frac{4}{3}(1+ \frac{3\gamma (N-1)}{2} - \frac{1}{2\tau} -\frac{1}{\tau(N-1)})-(\frac{1}{3} + \frac{5}{6 \tau}-\frac{4}{3\tau(N-1)})}{1+ \frac{3\gamma (N-1)}{2} - \frac{1}{2\tau} -\frac{1}{\tau(N-1)}} \\
\le \frac{\frac{4}{3}(1+ \frac{3\gamma (N-1)}{2} - \frac{1}{2\tau} -\frac{1}{\tau(N-1)})}{1+ \frac{3\gamma (N-1)}{2} - \frac{1}{2\tau} -\frac{1}{\tau(N-1)}} = \frac{4}{3}$ for each $N\ge 3$; because $\frac{1}{3}+\frac{5}{6 \tau} - \frac{4}{3\tau (N-1)}>0$. This bound is approached, for instance, when $\alpha$ and $\gamma$ are large and bigger than $\tau$: $K_N$ is the social optimum and $K_{1,N}$ is the worst-case Nash Equilibrium and the bounds in (\ref{socialCostUpperBound}) and the inequality for PoA are closely approached. If $\alpha < \gamma - \frac{1}{\tau (N-1)(N-2)}$, $K_N$ is a Nash Equilibrium and $\text{PoS}=1$, otherwise $\text{PoS}>1$.

\subsection*{Case $\alpha > 2\gamma - \frac{1}{\tau}$.}  
We first consider the links, whose deletion leaves the graph connected. For any node $i$, we focus on the links installed by $i$. Let $l=(i,j)$ be one such link and the number of all nodes $q$ that use $l$ as a link for the shortest paths from $j$ to $q$ is $z$. According to Schoone \emph{et al.}~\cite[Theorem 2.1., case $k=1$]{Schoone87}, all the distances from $i$ to the other nodes are increased by at most $2d$, where $d$ is the diameter in the original graph. In a Nash Equilibrium, $2d z \gamma - \alpha - \frac{1}{\tau d_i (d_i - 1)} >0$ for any possible value of $d_i \ge 2$, i.e. we obtain $2d z \gamma - \alpha - \frac{1}{2\tau} > 0$. Hence, $z > \frac{\alpha + \frac{1}{2\tau}}{2d \gamma}$ and then the number of such links to node $j$ is not bigger than $\frac{2d \gamma N}{\alpha + \frac{1}{2\tau}}$. Taking into account all possible nodes, the number of links whose deletion does not disconnect the graph is not bigger than $\frac{2d\gamma N^2}{\alpha + \frac{1}{2\tau}}$. On the other hand there are at most $(N-1)$ links such that a removal of any of those links disconnects the graph. Indeed, a connected graph has a spanning tree $T$ and a link removed from $T$ disconnects the graph, while a removal of a link that is not in $T$ leaves the graph connected. Therefore,
\small{
\begin{align}
L &\le N-1 + \frac{2d \gamma N^2}{\alpha + \frac{1}{2\tau}} \label{bound:NumberLinks}
\end{align}
}\normalsize If two nodes $i$ and $j$ are $a$ hops apart from each other, adding a link from $i$ to $j$ would reduce the hopcounts from $i$ to all the nodes in the ``second half'' along the previous path to $j$ by at least half of their lengths, by $a-1, a-3,\ldots, 1$ for $a$ even or by $a-1, a-3,\ldots, 2$ for $a$ odd. Hence, the total reduction in the sum of shortest paths from $i$ is $\sum_{i=1}^{\frac{a}{2}} (2i-1) = \frac{a^2}{4}$ or $\sum_{i=1}^{\frac{a-1}{2}} (2i) = \frac{a^2-1}{4}$ for $a$ even or odd, respectively. Assuming a Nash Equilibrium and $i$ is a starting node on the diameter, considering the change in cost $J_i$, the following inequality would hold for any $d$: $\alpha - \frac{d^2-1}{4} \gamma + \frac{1}{\tau d_i (d_i +1)}>0$. Using $d_i \ge 1$, and the absolute maximum for $d$ being $N-1$, we arrive at
\small{
\begin{align}
d &<\min \{\sqrt{1+\frac{4}{\gamma} (\alpha+\frac{1}{2 \tau})},N\} \label{bound:Diameter}
\end{align}
}\normalsize Each node $i$ has at least one neighbor and all the others are no more than $d$ hops apart, hence: $\sum_{j=1}^{N} h(i,j) \le 1 + (N-1) d$. Applying the arithmetic-harmonic mean inequality leads to $\frac{1}{\tau} \sum_{i=1}^N \frac{1}{d_i} \ge \frac{1}{\tau} \frac{N^2}{\sum_{i=1}^N d_i} = \frac{N^2}{2L \tau}$. We proceed by upper bounding $J$ in (\ref{socialCostWiener}),
\small{
\begin{align}
J & \le N+\alpha L + \gamma N(1 + (N-1) d) -\frac{N^2}{2L \tau} \label{upperBoundSocWelfare}
\end{align}
}\normalsize Applying (\ref{bound:NumberLinks}) modifies (\ref{upperBoundSocWelfare}), into
\footnotesize{
\begin{align}
J \le&
N+\alpha (N-1) + \gamma N + N\big( (N-1)+\frac{2\alpha N}{\alpha + \frac{1}{2\tau}}\big) \gamma d \notag\\
&-\frac{N^2}{2\tau(N-1 + \frac{2d \gamma N^2}{\alpha + \frac{1}{2\tau}}) }  \label{upperBoundApplyingL} 
\end{align}
}\normalsize

\begin{figure*}[h!tb]
\centering
\subfloat[Average Price of Anarchy (PoA) $\gamma = 5$.]{
\includegraphics[trim = 12.8mm 67mm 24.8mm 68mm,clip,width=0.32\textwidth]{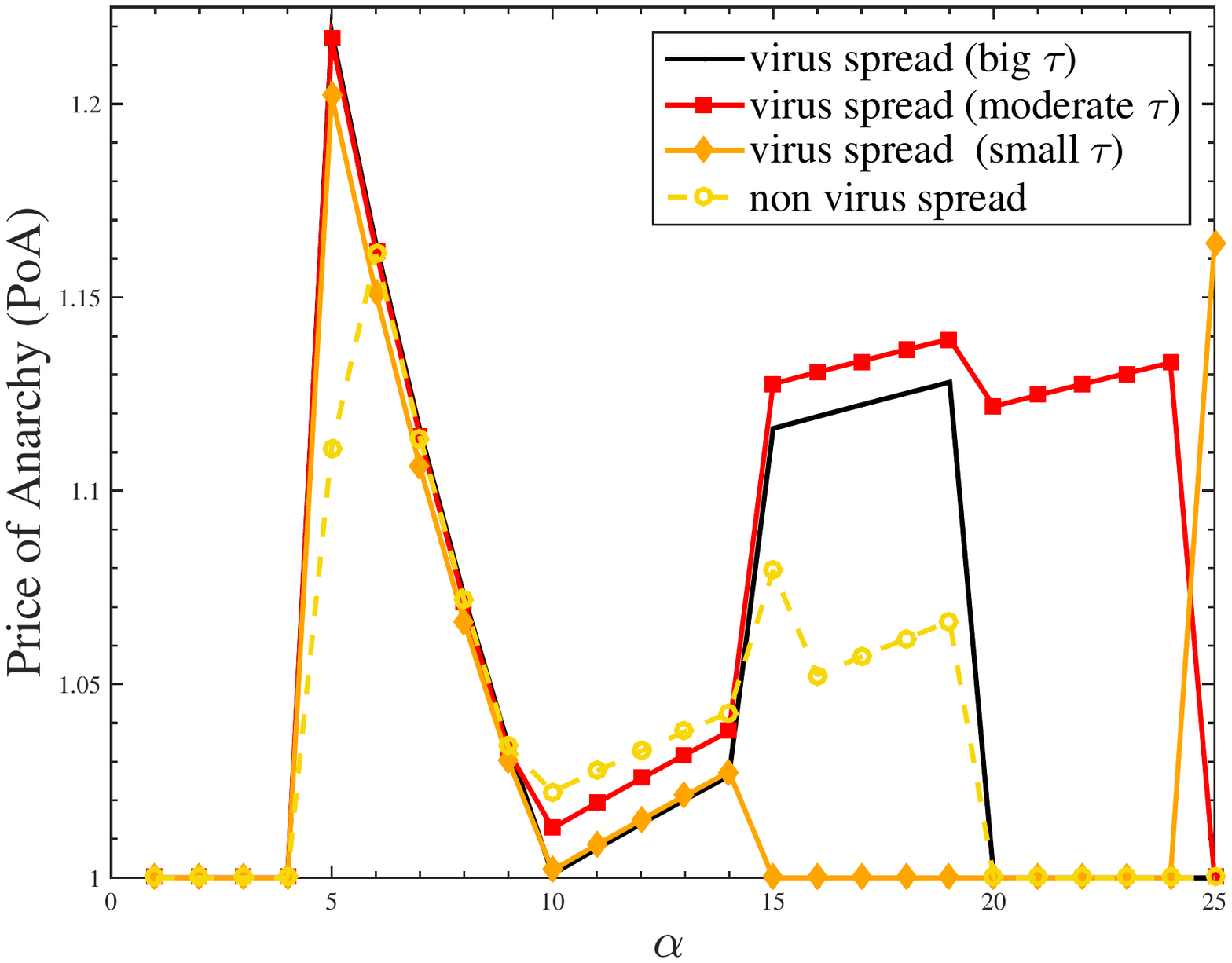}
\label{VSCNPoA}}
\subfloat[Average number of links $\gamma = 5$.]{
\includegraphics[trim = 12.8mm 67mm 24.8mm 68mm,clip,width=0.32\textwidth]{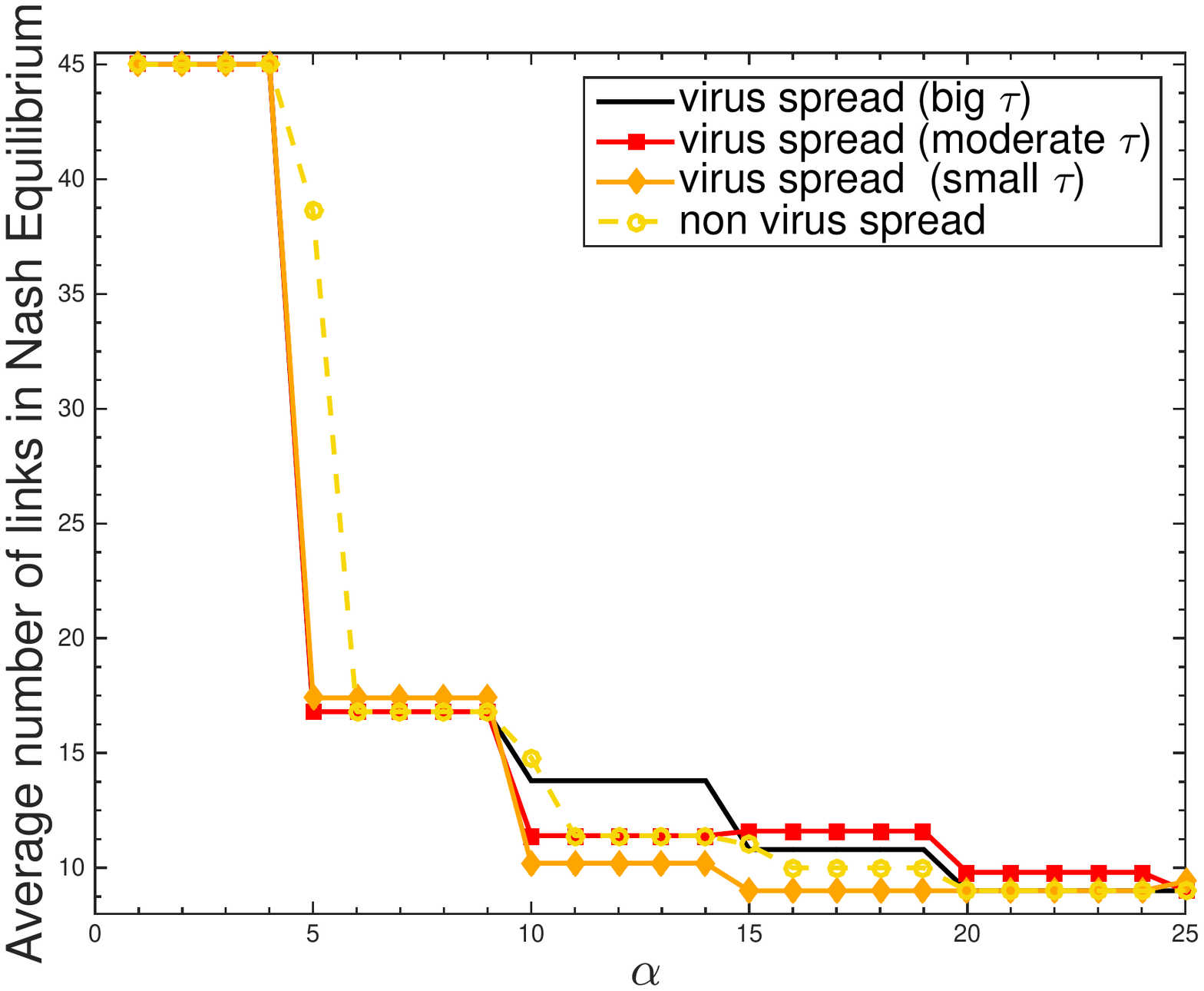}
\label{VSCNNumberLinks}}
\subfloat[Average hopcount $\gamma = 5$.]{
\includegraphics[trim = 12.8mm 67mm 24.8mm 68mm,clip,width=0.32\textwidth]{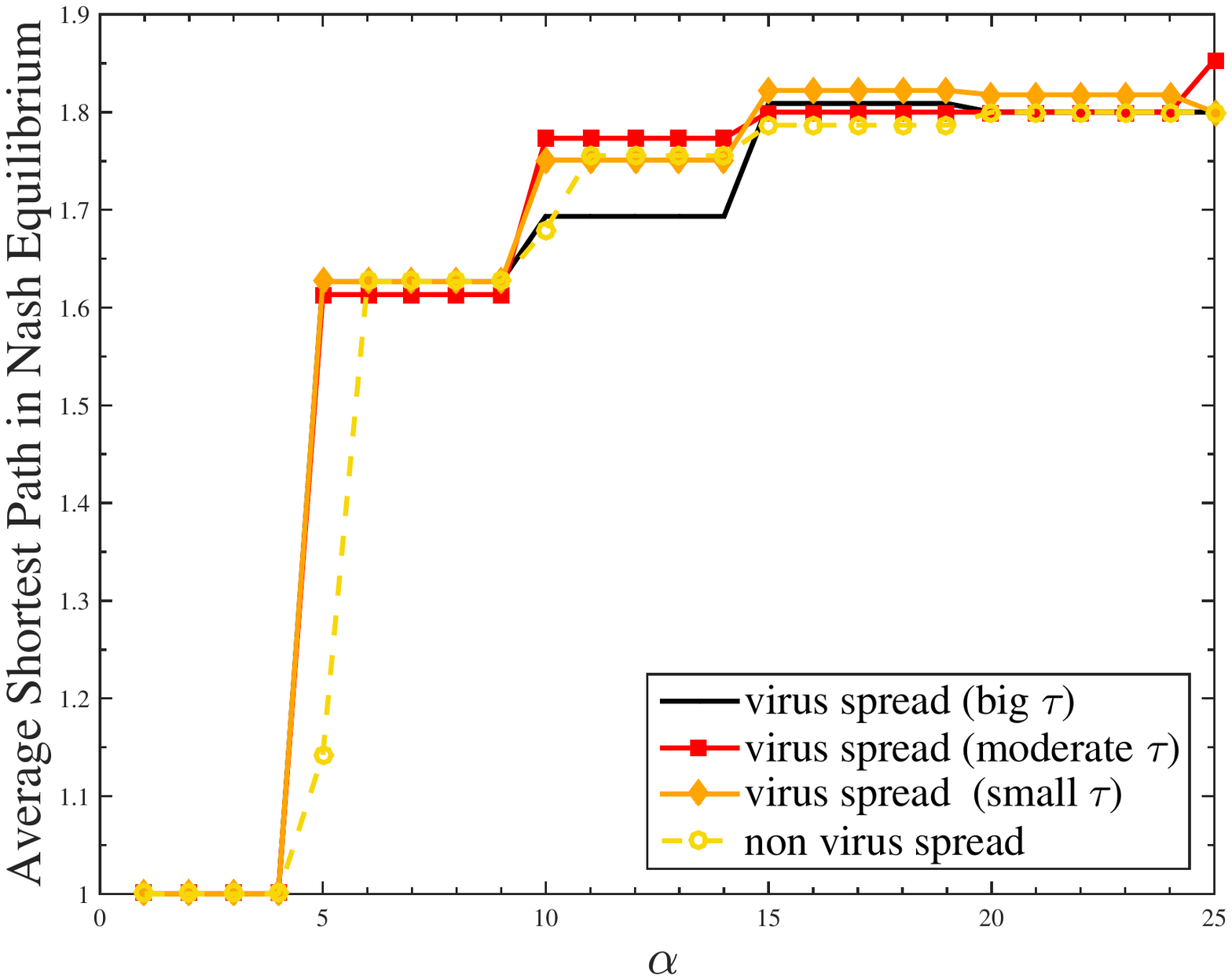}
\label{VSCNAvgShorestPath}}

\subfloat[Average Price of Anarchy (PoA) $\gamma = 1$.]{
\includegraphics[trim = 12.8mm 67mm 24.8mm 68mm,clip,width=0.32\textwidth]{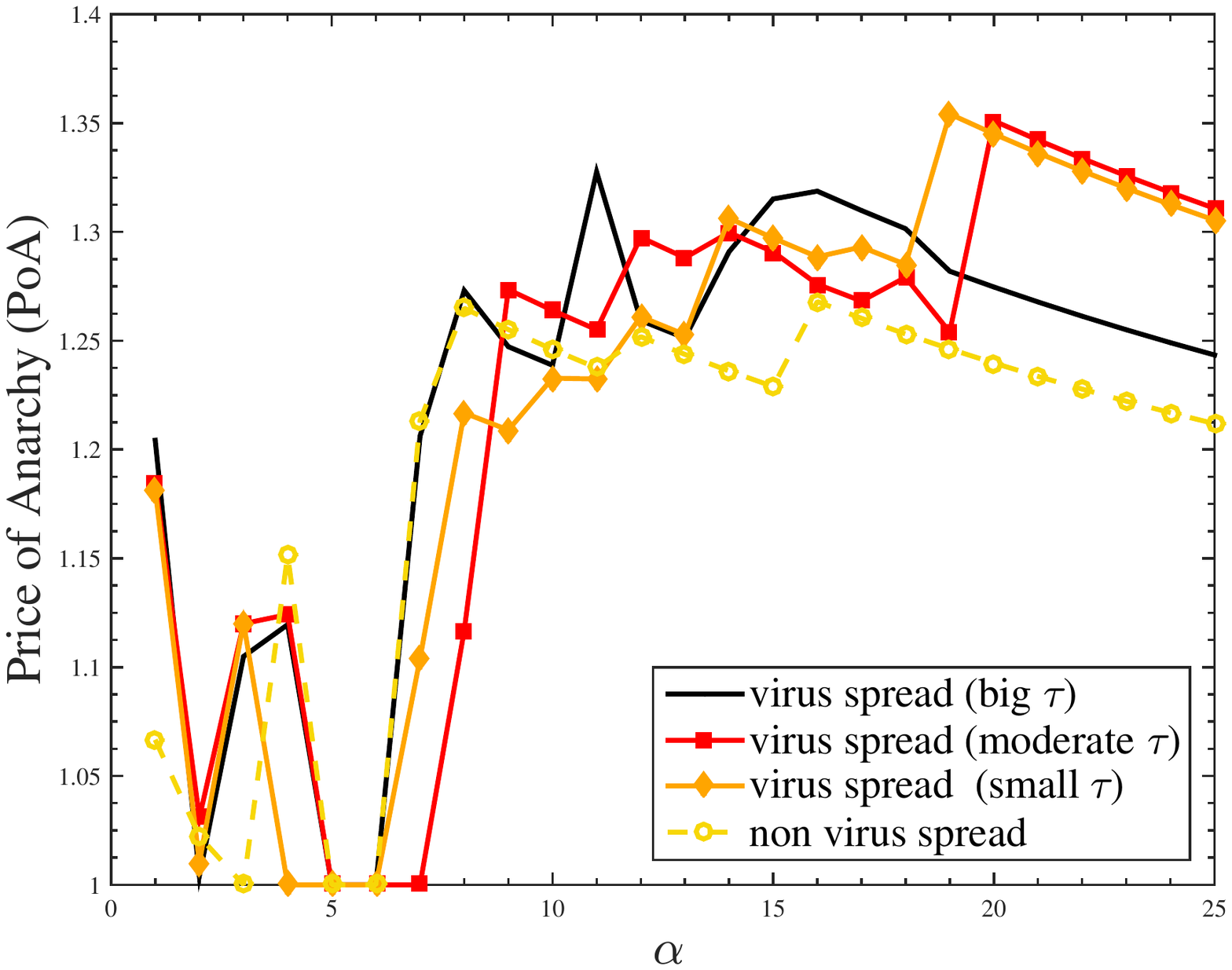}
\label{VSCNPoA1}}
\subfloat[Average number of links $\gamma = 1$.]{
\includegraphics[trim = 12.8mm 67mm 24.8mm 68mm,clip,width=0.32\textwidth]{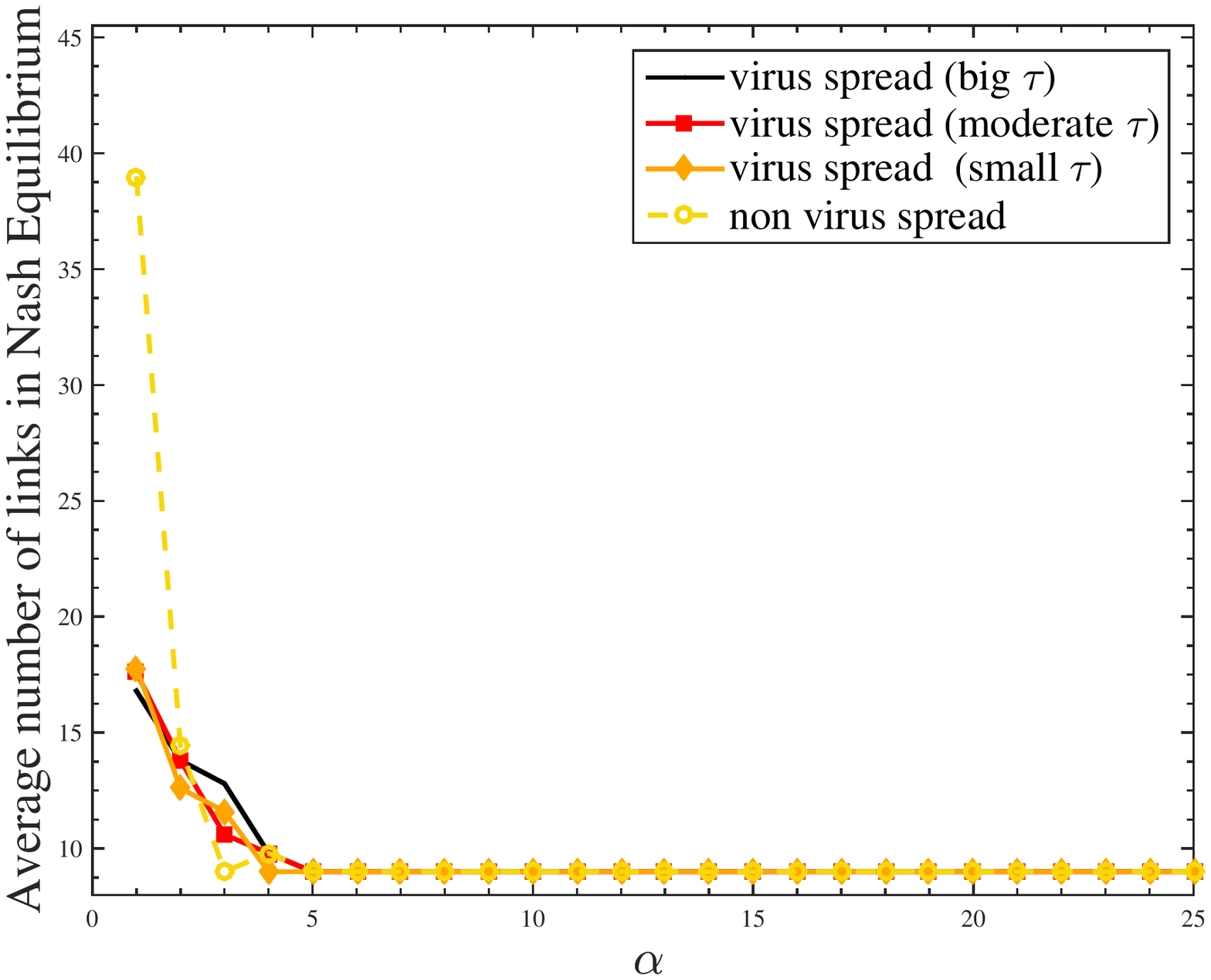}
\label{VSCNNumberLinks1}}
\subfloat[Average hopcount $\gamma = 1$.]{
\includegraphics[trim = 12.8mm 67mm 24.8mm 68mm,clip,width=0.32\textwidth]{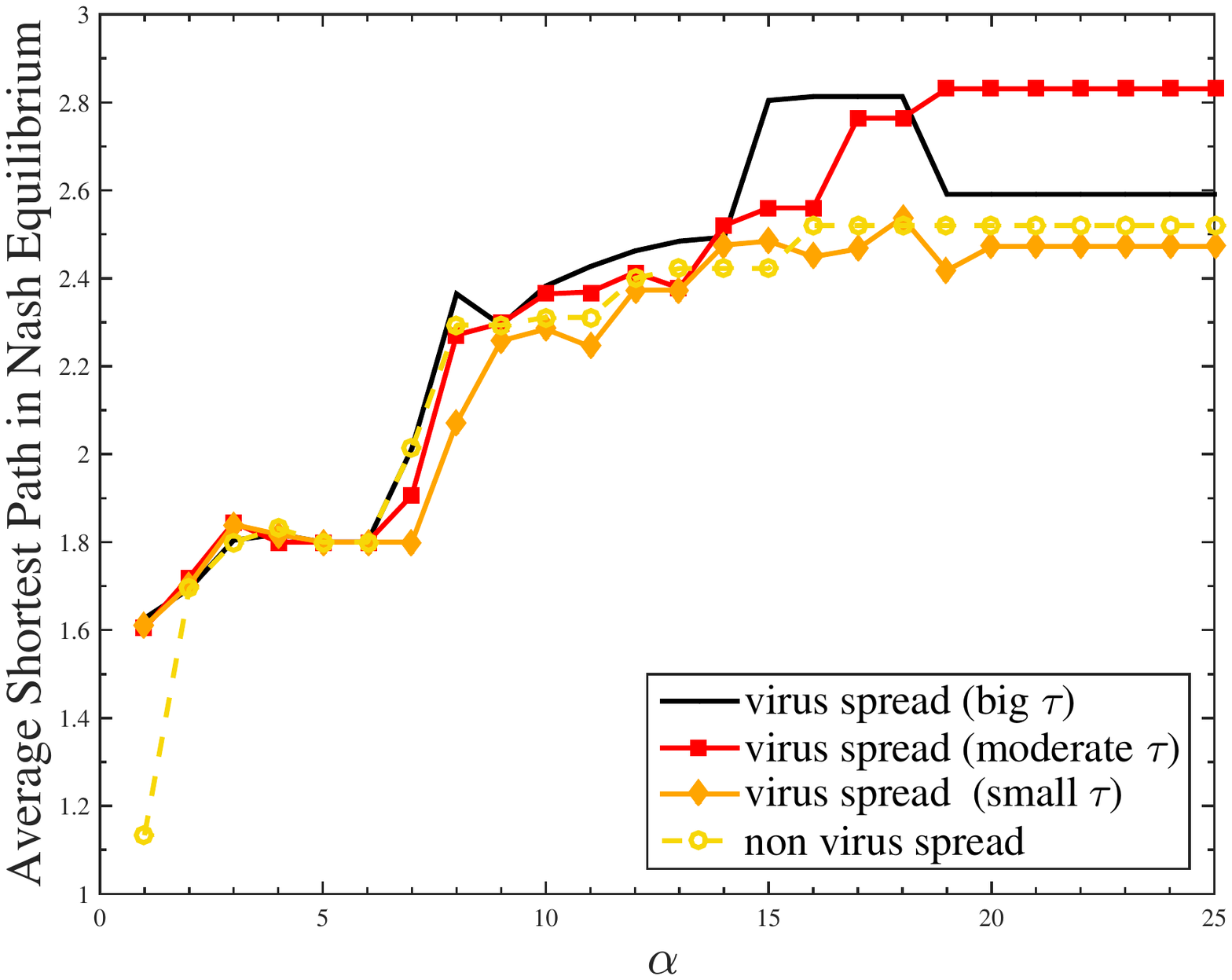}
\label{VSCNAvgShorestPath1}}

\subfloat[Average Price of Anarchy (PoA) $\gamma = 0.1$.]{
\includegraphics[trim = 12.8mm 67mm 24.8mm 68mm,clip,width=0.32\textwidth]{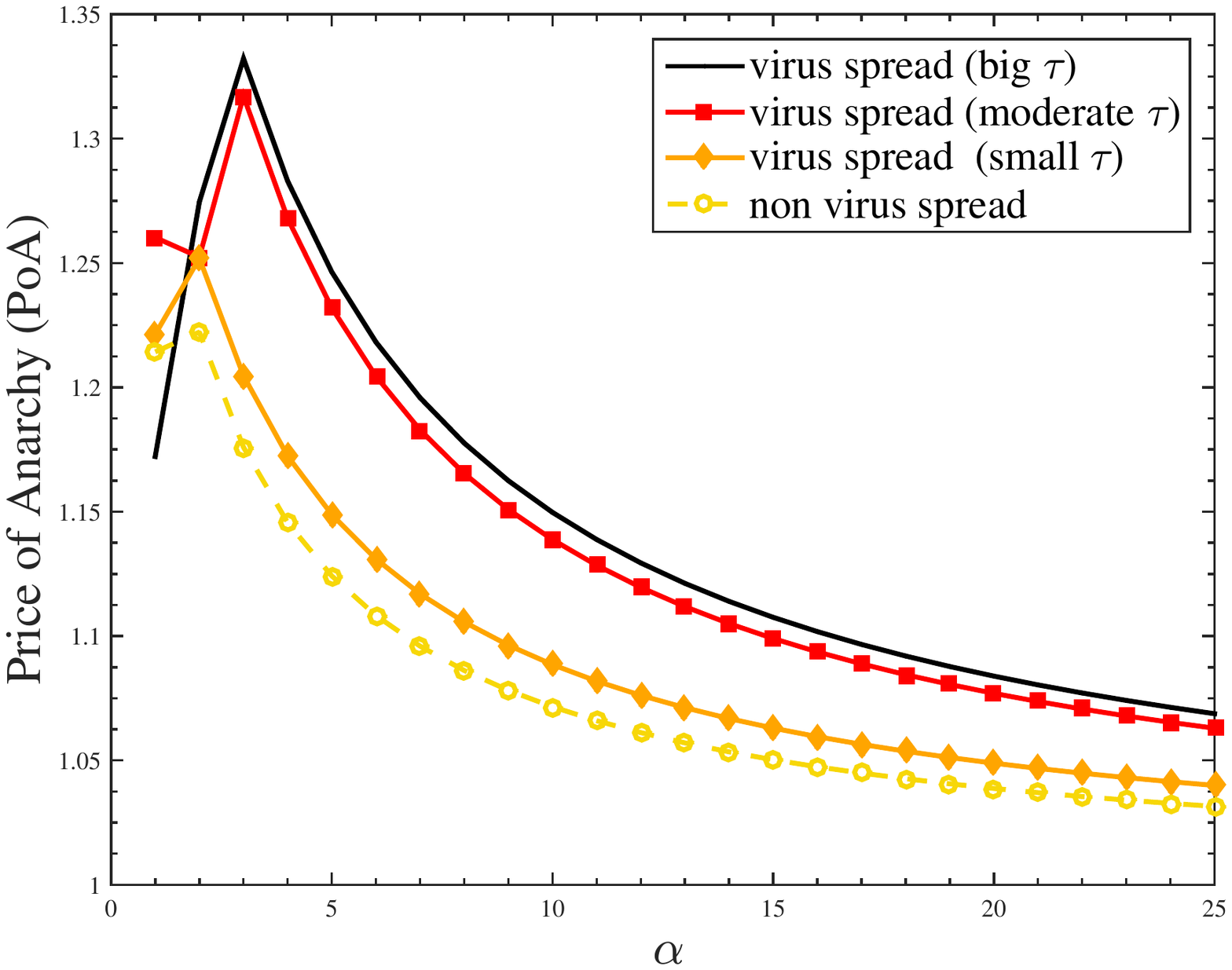}
\label{VSCNPoA01}}
\subfloat[Average number of links $\gamma = 0.1$.]{
\includegraphics[trim = 12.8mm 67mm 24.8mm 68mm,clip,width=0.32\textwidth]{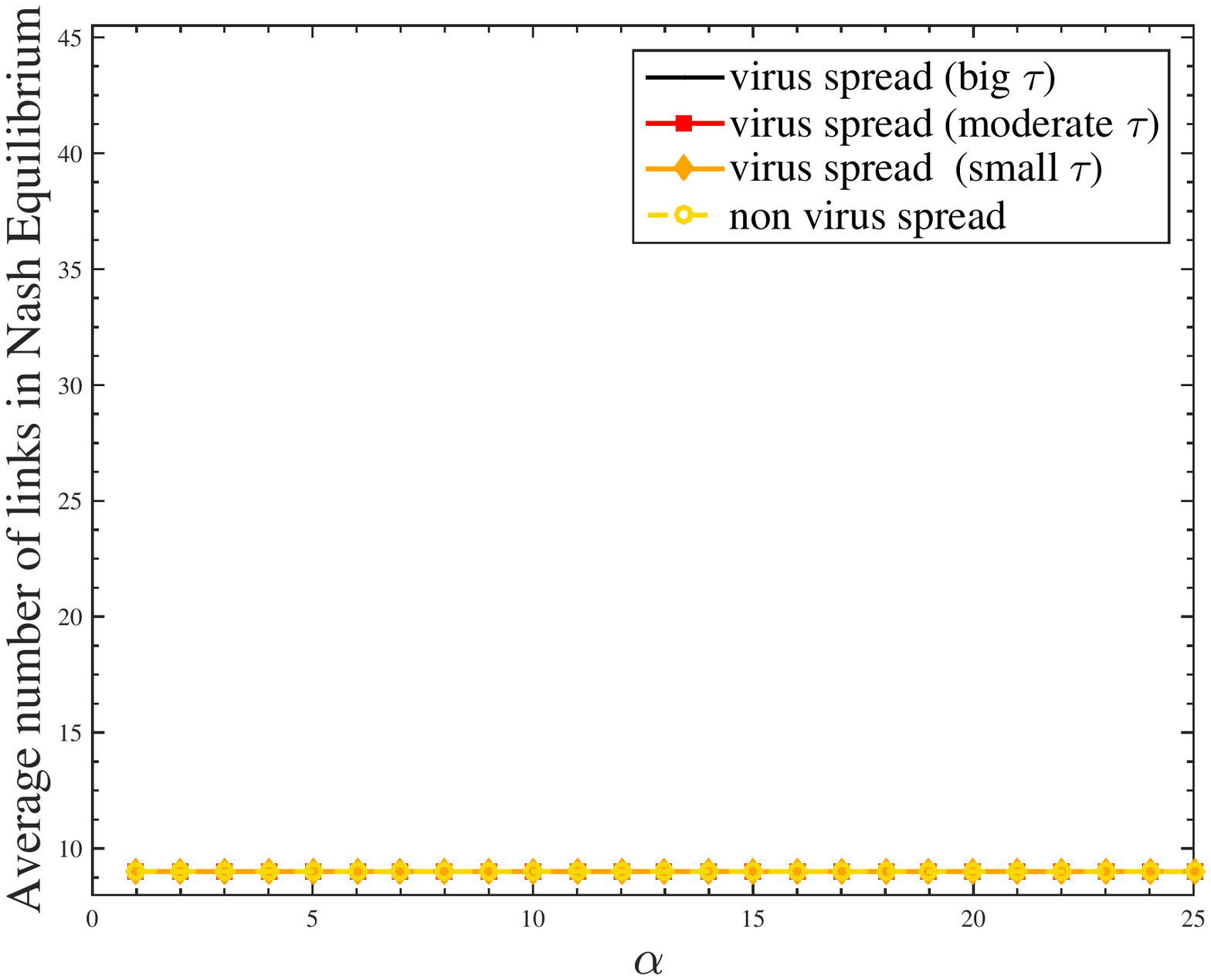}
\label{VSCNNumberLinks01}}
\subfloat[Average hopcount $\gamma = 0.1$.]{
\includegraphics[trim = 12.8mm 67mm 24.8mm 68mm,clip,width=0.32\textwidth]{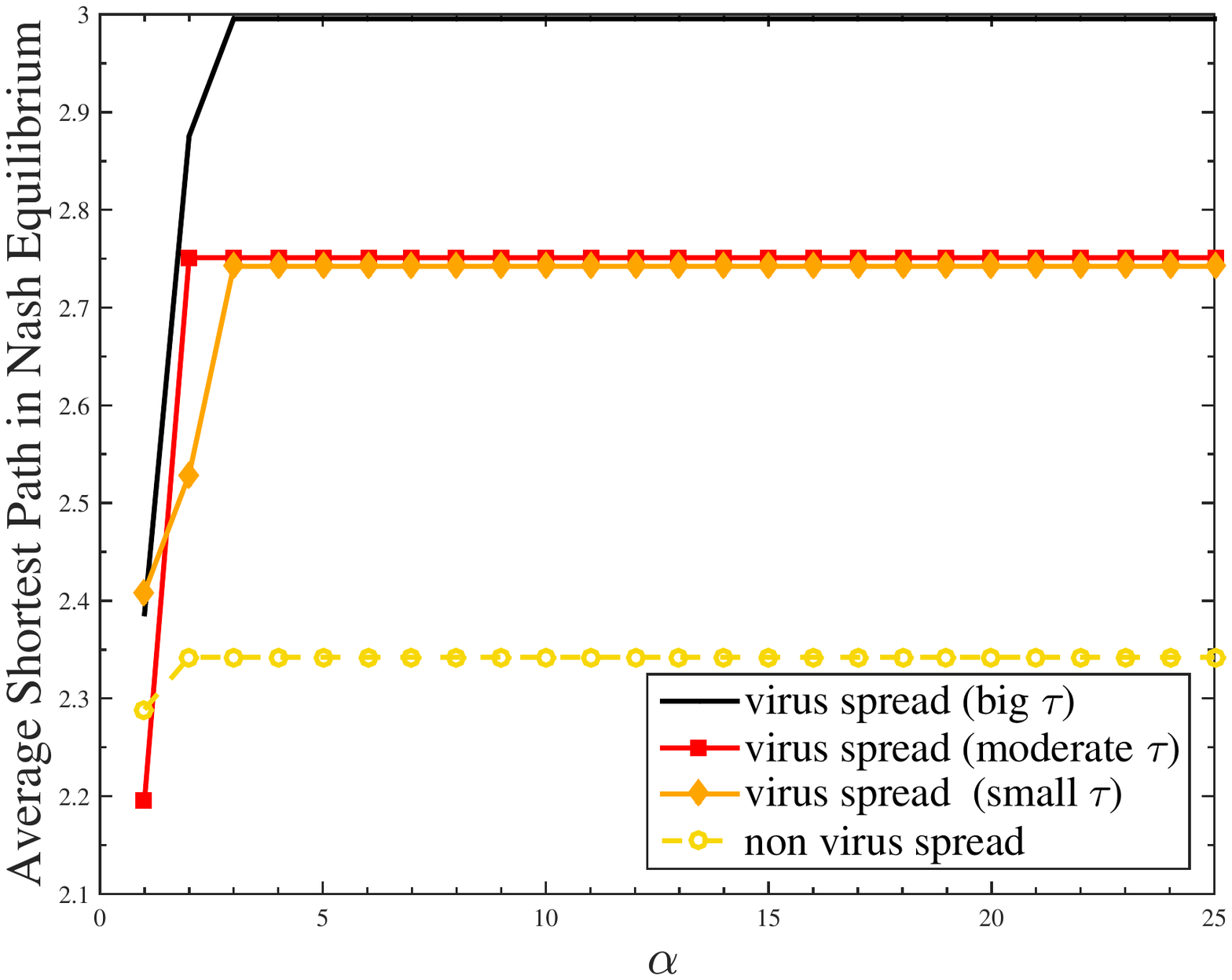}
\label{VSCNAvgShorestPath01}}

\caption[]{Simulation results of the heuristic algorithm for the obtained networks in a Nash Equilibrium. The three regimes big, moderate and small $\tau$ are represented with values $5.2$, $1.4$ and $1$, respectively. The number of nodes is $N=10$.} \label{fig:VSCNSimulation}
\end{figure*}

We distinguish two sub-cases, (a) and (b):

\hspace{-1em}\textbf{(a)} If $\alpha \ge 2\gamma - \frac{2}{\tau (N-1)}$, then the optimal social cost (and a Nash Equilibrium) is achieved for the star graph $K_{1,N-1}$, hence PoS=1. Now, using (\ref{upperBoundApplyingL}) for PoA,
\small{
\begin{align}
\text{PoA} \le& \frac{
N+\alpha L + \gamma N(1 + (N-1) d) -\frac{N^2}{2L \tau} }{N+ \alpha(N-1) + 2 \gamma (N-1)^2 - \frac{(N-1)^2+1}{\tau(N-1)}} \label{PoAtemp}
\end{align}
}\normalsize and applying (\ref{bound:NumberLinks}),
\footnotesize{
\begin{align}
\text{PoA} \le& \frac{N+\alpha (N-1) +  N \gamma + N \gamma(N-1+\frac{2  \alpha N}{\alpha + \frac{1}{2\tau}})d -\frac{N^2}{2\tau(N-1 + \frac{2d \gamma N^2}{\alpha + \frac{1}{2\tau}})}}{N+ \alpha(N-1) + 2 \gamma (N-1)^2 - \frac{(N-1)^2+1}{\tau(N-1)}}  \label{boundingPoA1} 
\end{align}
}\normalsize 

\begin{itemize}[noitemsep,nolistsep]
\item $\gamma \cdot d$ is not infinitesimally small. For sufficiently large\footnote{Significantly larger compared to the coefficients $\alpha$, $\gamma$ and $\tau$.} $N$, after some algebraic transformation, division by $N^2$ in both the numerator and denominator of (\ref{boundingPoA1}) and applying (\ref{bound:NumberLinks}), 
\small{
\begin{align}
\text{PoA} \le& O\big((\frac{1}{2}+\frac{\alpha}{\alpha + \frac{1}{2\tau}})  \sqrt{1+\frac{4}{\gamma} (\alpha+\frac{1}{2 \tau})} \big) \label{PoAgeneralbound}
\end{align}
}\normalsize 

\item $\gamma \cdot d$ is infinitesimally small. According to (\ref{bound:NumberLinks}), $L = O(N)$. Now, (\ref{PoAtemp}) yields
\small{
\begin{align}
\text{PoA} \le& O\Big(\frac{
N+\alpha L -\frac{N^2}{2L \tau} }{N+ \alpha(N-1) - \frac{(N-1)^2+1}{\tau(N-1)}}\Big)
\end{align}
}\normalsize
\end{itemize}

\hspace{-1em}\textbf{(b)} If $2\gamma - \frac{2}{\tau (N-1)} \ge \alpha \ge 2\gamma - \frac{1}{\tau}$, then the optimal social cost is achieved for the complete graph $K_{N}$, and using (\ref{upperBoundApplyingL}), $\text{PoA} \le \frac{
N+\alpha L + \gamma N(1 + (N-1) d) -\frac{N^2}{2L \tau} }{N+ \alpha \frac{N(N-1)}{2} + \gamma N(N-1) - \frac{N}{\tau(N-1)}}$. Now,

\begin{itemize}[noitemsep,nolistsep]
\item $\gamma \cdot d$ is not infinitesimally small. Using the bound for $L$ in (\ref{bound:NumberLinks}), for sufficiently high enough$^{\text{6}}$ $N$, (\ref{boundingPoA1}) is transformed into $\text{PoA} \le  O\Big( 2 \frac{\sqrt{\gamma^2+4\gamma (\alpha+\frac{1}{2 \tau})} }{ \alpha  + 2 \gamma} (1+ \frac{2\alpha}{\alpha + \frac{1}{2\tau}} ) \Big)$ and we have a constant value for PoA.

\item $\gamma \cdot d$ is infinitesimally small. Then $\alpha$ is small and PoA has a value close to $1$.
\end{itemize}

\hspace{-1.1em}Based on these results, we present Theorem~\ref{theor:PoANE_generalGames}.
\begin{theorem}\label{theor:PoANE_generalGames} For sufficiently high $\tau$ in the \textsc{VSPC} game, the PoA depends on the parameters $\alpha$, $\gamma$ and $\tau$,
\begin{enumerate}
\item if $\alpha \ge 2\gamma - \frac{2}{\tau (N-1)}$, then PoS=1 and if
\begin{itemize}[noitemsep,nolistsep]
\item $\gamma d$ is not small, then \\$\text{PoA} \le O\big((\frac{1}{2}+\frac{\alpha}{\alpha + \frac{1}{2\tau}})  \sqrt{1+\frac{4}{\gamma} (\alpha+\frac{1}{2 \tau})} \big) $.

\item $\gamma \cdot d$ is small, then PoA is given by Corollary~\ref{corr:PoSPoA}. 
\end{itemize} 

\item if $2\gamma - \frac{2}{\tau (N-1)} \ge \alpha \ge 2\gamma - \frac{1}{\tau}$ and
\begin{itemize}[noitemsep,nolistsep]
\item $\gamma d$ is not small, then \\$\text{PoA} \le O\Big( 2 \frac{\sqrt{\gamma^2+4\gamma (\alpha+\frac{1}{2 \tau})} }{ \alpha  + 2 \gamma} (1+ \frac{2\alpha}{\alpha + \frac{1}{2\tau}} ) \Big)$.

\item $\gamma \cdot d$ (and $\gamma$) is small, then $\alpha$ is also small and we have a constant value for PoA close to $1$.
\end{itemize} 

\item if $2\gamma - \frac{1}{\tau} > \alpha \ge \gamma - \frac{1}{\tau (N-1)} $, then Nash Equilibria graphs have diameters at most two and $\text{PoS} \le \text{PoA} < \frac{4}{3}$.

\item if $\gamma - \frac{1}{\tau (N-1)} > \alpha$, then $K_N$ is the only Nash Equilibrium and $\text{PoA}=\text{PoS}=1$.
\end{enumerate}
\end{theorem}
\hspace{-1.25em}Theorem~\ref{theor:PoANE_generalGames} and Corollaries~\ref{corr:PoSPoA} and \ref{corr:PoSPoA2} are compatible for small $\gamma$.

\subsection{Computational aspects and simulations}
Since, (i) for small $\tau$ below the epidemic threshold $\tau_c$, $v_{i \infty} =0$ and (ii) for $\tau = \infty$, $v_{i \infty} =1$ for all $i \in \{1,2,\ldots, N\}$, the problem of finding a best response in the \textsc{VSPC} game includes the best-response problem described in~\cite{Fabrikant:2003:NCG:872035.872088}, which is NP-hard.  We therefore use a best-response heuristic algorithm, as in~\cite{Chun2004}. The steps of the algorithm are the following:
\begin{enumerate}
\item We start with an initial random graph $G=G_{1}^1$.
\item Time $t$ is slotted and the first time slot is $t=1$.
\item Each node takes only two actions at each time slot $t$. We fix the order of actions from node $1$ to node $N$. The possible actions for each node are: dropping a link (D); adding a link (A); or doing nothing (N).
\item We denote by $G_t^i$ the graph at time $t$ before the action of node $i$.
\item Starting from node $1$, each node $i$ first computes the maximum reduction of its cost $J_i$ induced by dropping a link (D) from graph $G_t^{i}$, or takes action (N) if no reduction could be realized. Taking the obtained graph, node $i$ computes the maximum reduction of its cost $J_i$ induced by adding a link (A), or takes action (N) if no reduction could be realized.
\item After the decision of node $i$ at time $t$, the graph becomes $G_t^{i+1}$. After the decision of node $N$, the algorithm moves to time $t+1$ (i.e., to graph $G_{t+1}^{1}$).
\item An equilibrium is reached at time $t$ when all the nodes take the action (N) or the algorithm stops after a certain number of iterations $t_{\max}$ is reached.
\end{enumerate}

In Fig.~\ref{fig:VSCNSimulation}, results are given for the Price of Anarchy (PoA), the average number of links and the average hopcount as a function of installation cost $\alpha$ for different effective infection rates $\tau$ (namely big, moderate and small $\tau$ with values $5.2$, $1.4$ and $1$, respectively that well represents the 3 regimes) and different weights (costs) for the hopcounts $\gamma$ in a graph with $N=10$ nodes. We also display some typical outcomes of the algorithm for different values of $\alpha$ and $\tau$ and three different values of $\gamma$ in Fig.~\ref{fig:VizTopol2Gamma5}, Fig.~\ref{fig:VizTopol2Gamma1} and Fig.~\ref{fig:VizTopol2Gamma01} (see the visualizations at the end of the paper, after the bibliography). For all the metrics shown in Fig.~\ref{fig:VSCNSimulation}, there is an interesting behavior for the curve with ``no virus,'' in the sense that it follows the same shape as the curves where the virus is present, but is often shifted/delayed from them. This is due to the ''enhancing'' effect from the virus spread on the installation cost contribution. 

For small values of $\alpha$, due to the resulting cheap installation cost, and for non-negligible performance values $\gamma$, the  Nash Equilibrium is a very dense graph, often the complete graph $K_N$, for all $\tau$ (Fig.~\ref{fig:VizTopol2Gamma5}). This reflects case 4) in Theorem~\ref{theor:PoANE_generalGames}, although the interval for $\alpha$, where $K_N$ is the only Nash Equilibrium would shrink (and may vanish) for small $\gamma$ (Fig.~\ref{fig:VizTopol2Gamma1} and Fig.~\ref{fig:VizTopol2Gamma01}). In the latter case, the corresponding PoAs in Fig.~\ref{VSCNPoA01} and Fig.~\ref{fig:VSCN10} have comparable shapes and the obtained topologies are trees (see Fig.~\ref{VSCNNumberLinks01}), although not necessarily star graphs (see Fig.~\ref{fig:VizTopol2Gamma01}). 

Because of the higher installation cost (higher $\alpha$) in intervals 3) and 2) from Theorem~\ref{theor:PoANE_generalGames}, the Nash Equilibria topologies are sparse: (i) in particular non-tree networks for $\gamma = 5$ (constant average number of links and sum of hopcounts) and mostly trees for $\gamma$ equal to $1$ or $0.1$. Consequently, the PoA linearly decreases with $\alpha$ on this interval. For the interval 1) in Theorem~\ref{theor:PoANE_generalGames} (high installation cost), the PoA increases with $\alpha$ (and $\tau$), reaching a local maximum and then have a different behavior for larger $\alpha$. Namely, it decreases towards $1$ for small $\gamma$, while it is unpredictable for higher $\gamma$ (the left column, subfigures (a), (d) and (g), in Fig.~\ref{fig:VSCNSimulation}). But most importantly the effect of the epidemics part is noticeable and higher $\tau$ introduces inefficiency, which is reflected by a high PoA. It is also important to note that for comparable $\alpha$ and $\gamma$, the algorithm displays somewhat fluctuating behavior in terms of PoA (middle row of Fig.~\ref{fig:VSCNSimulation}) due to the heuristic nature of the algorithm.

Being able to detect the intervals with high PoA, as we have done in this section, means that for those intervals some coordination/incentives of and for the players is needed. Since the PoS is generally low, in the best case with Nash Equilibrium a small amount of coordination likely suffices. For the intervals where the PoA is low, the selfish behavior of the players might still lead to (near)-optimal topologies without any coordination.

\section{Related work}\label{sec:RelatedWork}
Virus spread in networks has been thoroughly explored during the last decades~\cite{Omic09,Chakrabarti2008,Ganesh2005,VanMieghem2011,RevModPhys.87.925}. These works involve studies ranging from virus-spread propagation, the computation of the number of infected hosts~\cite{Omic09} to the epidemic threshold~\cite{Chakrabarti2008} in various epidemic models on networks. There is a large body of literature on game formation, that mostly minimizes a cost utility based on hopcount and the cost for installing links~\cite{Fabrikant:2003:NCG:872035.872088,Corbo2005,Albers:2006:NEN:1109557.1109568,Chun2004,Moscibroda:2006:TFS:1146381.1146403,NisanRoughgardenTardosVazirani(AlgoGameTheory)07,5062080}. Fabrikant \emph{et al.}~\cite{Fabrikant:2003:NCG:872035.872088} have studied the case, where a node's utility is a weighted sum of the installed links and the sum of hopcounts from each node in an undirected graph. The follow up work by Albers \emph{et al.}~\cite{Albers:2006:NEN:1109557.1109568} resolved some open questions from~\cite{Fabrikant:2003:NCG:872035.872088}. Chun \emph{et al.}~\cite{Chun2004} have conducted extensive simulations on the same type of game formation. A game formation problem involving hopcounts and costs, applied to P2P networks has been considered by Moscibroda \emph{et al.}~\cite{Moscibroda:2006:TFS:1146381.1146403}. Meirom \emph{et al.}~\cite{Meirom:2014:NFG:2600057.2602862,7218557} have provided dynamic and data analyses (apart from their static analysis) in an NFG setting with heterogenous players and robustness objectives. Nahir \emph{et al.}~\cite{5062080} have considered similar NFG problems in directed graphs. A coalition and bilateral agreements between players in NFG and game-theory, in general, have been considered in~\cite{Avrachenkov2011,Iosifidis2014,Yerramalli2014,Corbo2005}. In order to evaluate ``the goodness'' of the equilibria, the prices of anarchy and stability~\cite{NisanRoughgardenTardosVazirani(AlgoGameTheory)07,Koutsoupias2009} have been used. 

In this work, we have considered virus protection aspects together with cost and the length of shortest paths. In this sense, our work extends (with virus spread) and generalizes the related work~\cite{Fabrikant:2003:NCG:872035.872088,Albers:2006:NEN:1109557.1109568,5062080}. However, to the best of our knowledge, network formation games concerning virus spread and protection both with or without the performance aspects have not been considered in the NFG framework, although \emph{security games}~\cite{5062065,Aspnes2006,TechReport_Acemoglu2013,Lelarge2008,TCNSSTrajanovskiprotectionEpidemics,CDC2014_SISProtection} have been used in modeling the virus spread suppression and network immunization. Performance aspects, represented by the hopcounts are linearly independent from the resilience to virus spread-- the two metrics do not possess closed-form expressions-- making the NFG problem challenging, apart from the novelty.



\section{Conclusion}\label{sec:Conclusion}
We have considered a novel network formation game (NFG), called the \emph{virus spread-performance-cost} (\textsc{VSPC}) game,  for communication networks in which the aspects link installation costs, virus infection probability, and performance in terms of the number of hops needed to reach other nodes in the network, all need to be balanced. We have characterized the Nash Equilibria and the Price of Anarchy (PoA) for various cases. In most of the cases, the PoA is not high, often close to $1$, which implies that the decisions of non-cooperative players would lead, in a decentralized way, to an optimal topology.
%

When the aspect of the shortest hopcounts is not important, we have found that only \emph{trees} (but not all) could be Nash Equilibria. In that case, surprisingly, a \emph{path graph} is the worst- and the \emph{star graph} is the best-case Nash Equilibrium for big virus infection rate $\tau$, while it is the opposite for small $\tau$. For intermediate values of $\tau$, other trees are optimal. The PoA is the highest for values of $\tau$ just above the epidemic threshold. However, the PoA is generally small and close to $1$, does not depend on the number of players, and is inversely proportional to $\tau$ and the installation cost $\alpha$. 

When the hopcounts do matter, the Nash Equilibria might be formed by more complex topologies. The PoA highly depends on $\tau$, the installation links and hopcount costs $\alpha$ and $\gamma$, respectively, as shown by both theory and simulation. Although the PoA is small for most of the cases, for some intervals of those parameters, the PoA could be high. Hence, a central control and regulatory mechanism should be in place in such cases. Being able to detect those intervals, as we have done, helps in the design of optimal, efficient, virus-free and cheap overlay, P2P or wireless networks by limiting the non-cooperative freedom of the hosts' decisions.

There are several possibilities for follow-up work, such as a study on mixed Equilibria, player coalitions, inhomogeneous costs, or time-varying networks.

\newcommand{\BIBdecl}{\setlength{\itemsep}{0.1 em}}

\bibliographystyle{ieeetran}

\begin{thebibliography}{10}
\providecommand{\url}[1]{#1}
\csname url@samestyle\endcsname
\providecommand{\newblock}{\relax}
\providecommand{\bibinfo}[2]{#2}
\providecommand{\BIBentrySTDinterwordspacing}{\spaceskip=0pt\relax}
\providecommand{\BIBentryALTinterwordstretchfactor}{4}
\providecommand{\BIBentryALTinterwordspacing}{\spaceskip=\fontdimen2\font plus
\BIBentryALTinterwordstretchfactor\fontdimen3\font minus
  \fontdimen4\font\relax}
\providecommand{\BIBforeignlanguage}[2]{{%
\expandafter\ifx\csname l@#1\endcsname\relax
\typeout{** WARNING: IEEEtran.bst: No hyphenation pattern has been}%
\typeout{** loaded for the language `#1'. Using the pattern for}%
\typeout{** the default language instead.}%
\else
\language=\csname l@#1\endcsname
\fi
#2}}
\providecommand{\BIBdecl}{\relax}
\BIBdecl

\bibitem{CDC2015_GameFormationVirusSpread}
S.~Trajanovski, F.~{A. Kuipers}, Y.~Hayel, E.~Altman, and P.~{Van Mieghem},
  ``Designing virus-resistant networks: a game-formation approach,'' in
  \emph{Proc. of IEEE CDC (54th IEEE Conference on Decision and
  Control)}.\hskip 1em plus 0.5em minus 0.4em\relax Osaka, Japan: IEEE,
  {December} 2015, pp. 294--299.

\bibitem{neglia07}
G.~Neglia, G.~Lo~Presti, H.~Zang, and D.~Towsley, ``A network formation game
  approach to study bittorrent tit-for-tat,'' in \emph{NetgCoop}, 2007.

\bibitem{Corbo2005}
J.~Corbo and D.~Parkes, ``The price of selfish behavior in bilateral network
  formation,'' in \emph{PODC}.\hskip 1em plus 0.5em minus 0.4em\relax ACM,
  2005, pp. 99--107.

\bibitem{Chun2004}
B.-G. Chun, R.~Fonseca, I.~Stoica, and J.~Kubiatowicz, ``Characterizing
  selfishly constructed overlay routing networks,'' in \emph{IEEE INFOCOM},
  vol.~2, March 2004, pp. 1329--1339.

\bibitem{Shakkottai2007}
S.~Shakkottai, E.~Altman, and A.~Kumar, ``Multihoming of users to access points
  in wlans: A population game perspective,'' \emph{IEEE Jour. on Sel. Areas in
  Communications}, vol.~25, no.~6, pp. 1207--1215, 2007.

\bibitem{Krunz2014}
M.~K. Hanawal, M.~J. Abdel-Rahman, and M.~Krunz, ``Game theoretic anti-jamming
  dynamic frequency hopping and rate adaptation in wireless systems,'' in
  \emph{WiOpt}, 2014, pp. 247--254.

\bibitem{mackenzie2006game}
A.~MacKenzie and L.~DaSilva, \emph{Game Theory for Wireless Engineers}, ser.
  Synth. lect. on communications.\hskip 1em plus 0.5em minus 0.4em\relax Morgan
  \& Claypool Publ., 2006.

\bibitem{Iosifidis2014}
G.~Iosifidis, L.~Gao, J.~Huang, and L.~Tassiulas, ``Enabling crowd-sourced
  mobile internet access,'' in \emph{INFOCOM}, 2014, pp. 451--459.

\bibitem{Watanabe2008}
E.~Watanabe, D.~Menasch\'{e}, E.~{de Souza e Silva}, and R.~Le\~{a}o,
  ``Modeling resource sharing dynamics of voip users over a wlan using a
  game-theoretic approach,'' in \emph{INFOCOM}, 2008, pp. 1588--1596.

\bibitem{Jiang2013}
B.~Jiang, N.~Hegde, L.~Massoulie, and D.~Towsley, ``How to optimally allocate
  your budget of attention in social networks,'' in \emph{INFOCOM}.\hskip 1em
  plus 0.5em minus 0.4em\relax IEEE, April 2013, pp. 2373--2381.

\bibitem{Marbach2011}
F.~Ming, F.~Wong, and P.~Marbach, ``Who are your friends? - a simple mechanism
  that achieves perfect network formation,'' in \emph{INFOCOM}.\hskip 1em plus
  0.5em minus 0.4em\relax IEEE, 2011, pp. 566--570.

\bibitem{Fabrikant:2003:NCG:872035.872088}
A.~Fabrikant, A.~Luthra, E.~Maneva, C.~H. Papadimitriou, and S.~Shenker, ``On a
  network creation game,'' in \emph{PODC}.\hskip 1em plus 0.5em minus
  0.4em\relax ACM, 2003, pp. 347--351.

\bibitem{Aumann88}
R.~Aumann and R.~Myerson, \emph{Endogenous Formation of Links Between Players
  and of Coalitions: an Application of the Shapley Value}.\hskip 1em plus 0.5em
  minus 0.4em\relax Stanford University, 1988.

\bibitem{Omic09}
P.~Van~Mieghem, J.~Omi\'{c}, and R.~Kooij, ``Virus spread in networks,''
  \emph{IEEE/ACM Trans. on Netw.}, vol.~17, no.~1, pp. 1 --14, feb. 2009.

\bibitem{Chakrabarti2008}
D.~Chakrabarti, Y.~Wang, C.~Wang, J.~Leskovec, and C.~Faloutsos, ``Epidemic
  thresholds in real networks,'' \emph{ACM Trans. Inf. Syst. Secur.}, vol.~10,
  no.~4, pp. 1:1--1:26, Jan. 2008.

\bibitem{Ganesh2005}
A.~Ganesh, L.~Massoulie, and D.~Towsley, ``The effect of network topology on
  the spread of epidemics,'' in \emph{IEEE INFOCOM}, vol.~2, 2005, pp.
  1455--1466.

\bibitem{VanMieghem2011}
P.~Van~Mieghem, ``\BIBforeignlanguage{English}{The {N-}intertwined {SIS}
  epidemic network model},'' \emph{\BIBforeignlanguage{English}{Computing}},
  vol.~93, no. 2-4, pp. 147--169, 2011.

\bibitem{PhysRevE.91.032812}
P.~Van~Mieghem and R.~van~de Bovenkamp, ``Accuracy criterion for the mean-field
  approximation in susceptible-infected-susceptible epidemics on networks,''
  \emph{Phys. Rev. E}, vol.~91, p. 032812, Mar 2015.

\bibitem{Albers:2006:NEN:1109557.1109568}
S.~Albers, S.~Eilts, E.~Even-Dar, Y.~Mansour, and L.~Roditty, ``On nash
  equilibria for a network creation game,'' in \emph{SODA}, 2006, pp. 89--98.

\bibitem{PAVanMieghem2}
P.~Van~Mieghem, \emph{Graph Spectra of Complex Networks}.\hskip 1em plus 0.5em
  minus 0.4em\relax Cambridge University Press, UK, 2011.

\bibitem{LovaszPelikan1973}
L.~Lov\'{a}sz and J.~Pelik\'{a}n, ``\BIBforeignlanguage{English}{On the
  eigenvalues of trees},'' \emph{\BIBforeignlanguage{English}{Periodica
  Mathematica Hungarica}}, vol.~3, no. 1-2, pp. 175--182, 1973.

\bibitem{Cioaba20061959}
S.~{M. Cioab\u{a}}, ``Sums of powers of the degrees of a graph,''
  \emph{Discrete Mathematics}, vol. 306, no.~16, pp. 1959 -- 1964, 2006.

\bibitem{Schoone87}
A.~A. Schoone, H.~L. Bodlaender, and J.~{van Leeuwen}, ``Diameter increase
  caused by edge deletion,'' \emph{Journal of Graph Theory}, vol.~11, no.~3,
  pp. 409--427, 1987.

\bibitem{RevModPhys.87.925}
R.~Pastor-Satorras, C.~Castellano, P.~{Van Mieghem}, and A.~Vespignani,
  ``Epidemic processes in complex networks,'' \emph{Rev. Mod. Phys.}, vol.~87,
  pp. 925--979, Aug 2015.

\bibitem{Moscibroda:2006:TFS:1146381.1146403}
T.~Moscibroda, S.~Schmid, and R.~Wattenhofer, ``On the topologies formed by
  selfish peers,'' in \emph{PODC}.\hskip 1em plus 0.5em minus 0.4em\relax ACM,
  2006, pp. 133--142.

\bibitem{NisanRoughgardenTardosVazirani(AlgoGameTheory)07}
N.~Nisan, T.~Roughgarden, E.~Tardos, and V.~{V. Vazirani}, \emph{Algorithmic
  Game Theory}.\hskip 1em plus 0.5em minus 0.4em\relax Cambridge University
  Press, 2007.

\bibitem{5062080}
A.~Nahir, A.~Orda, and A.~Freund, ``Topology design and control: A
  game-theoretic perspective,'' in \emph{IEEE INFOCOM}, april 2009, pp. 1620
  --1628.

\bibitem{Meirom:2014:NFG:2600057.2602862}
E.~A. Meirom, S.~Mannor, and A.~Orda, ``Network formation games with
  heterogeneous players and the internet structure,'' in \emph{ACM EC
  (Conference on Economics and Computation)}.\hskip 1em plus 0.5em minus
  0.4em\relax New York, NY, USA: ACM, 2014, pp. 735--752.

\bibitem{7218557}
------, ``Formation games of reliable networks,'' in \emph{IEEE INFOCOM}, April
  2015, pp. 1760--1768.

\bibitem{Avrachenkov2011}
K.~Avrachenkov, J.~Elias, F.~Martignon, G.~Neglia, and L.~Petrosyan,
  ``\BIBforeignlanguage{English}{A nash bargaining solution for cooperative
  network formation games},'' in \emph{\BIBforeignlanguage{English}{IFIP
  Networking}}, ser. LNCS.\hskip 1em plus 0.5em minus 0.4em\relax Springer,
  2011, vol. 6640, pp. 307--318.

\bibitem{Yerramalli2014}
S.~Yerramalli, R.~Jain, and U.~Mitra, ``Coalitional games for transmitter
  cooperation in mimo multiple access channels,'' \emph{IEEE Trans. on Signal
  Processing}, vol.~62, no.~4, pp. 757--771, Feb 2014.

\bibitem{Koutsoupias2009}
E.~Koutsoupias and C.~Papadimitriou, ``Worst-case equilibria,'' \emph{Computer
  Science Review}, vol.~3, no.~2, pp. 65 -- 69, 2009.

\bibitem{5062065}
J.~Omi\'{c}, A.~Orda, and P.~Van~Mieghem, ``Protecting against network
  infections: A game theoretic perspective,'' in \emph{IEEE INFOCOM}, 2009, pp.
  1485--1493.

\bibitem{Aspnes2006}
J.~Aspnes, K.~Chang, and A.~Yampolskiy, ``Inoculation strategies for victims of
  viruses and the sum-of-squares partition problem,'' \emph{Jour. of Computer
  and System Sciences}, vol.~72, no.~6, pp. 1077 -- 1093, 2006.

\bibitem{TechReport_Acemoglu2013}
D.~Acemoglu, A.~Malekian, and A.~Ozdaglar, ``Network security and contagion,''
  \emph{Journal of Economic Theory}, vol. 166, pp. 536 -- 585, 2016.

\bibitem{Lelarge2008}
M.~Lelarge and J.~Bolot, ``Network externalities and the deployment of security
  features and protocols in the internet,'' in \emph{SIGMETRICS}.\hskip 1em
  plus 0.5em minus 0.4em\relax ACM, 2008, pp. 37--48.

\bibitem{TCNSSTrajanovskiprotectionEpidemics}
S.~Trajanovski, Y.~Hayel, E.~Altman, H.~Wang, and P.~{Van Mieghem},
  ``{Decentralized Protection Strategies against SIS Epidemics in Networks},''
  \emph{IEEE Trans. Control of Network Systems}, vol.~2, pp. 406--419, 2015.

\bibitem{CDC2014_SISProtection}
Y.~Hayel, S.~Trajanovski, E.~Altman, H.~Wang, and P.~{Van Mieghem}, ``Complete
  game{-}theoretic characterization of sis epidemics protection strategies,''
  in \emph{Proc. of IEEE CDC (53rd IEEE Conference on Decision and Control)},
  Los Angeles, CA, USA, {December} 2014, pp. 1179--1184.

\end{thebibliography}


\linespread{0.7}

\captionsetup[subfigure]{labelformat=empty}
\begin{figure*}[p]

\centering
\includegraphics[trim = 0mm 0mm 0mm 0mm,clip,width=1.0\textwidth]{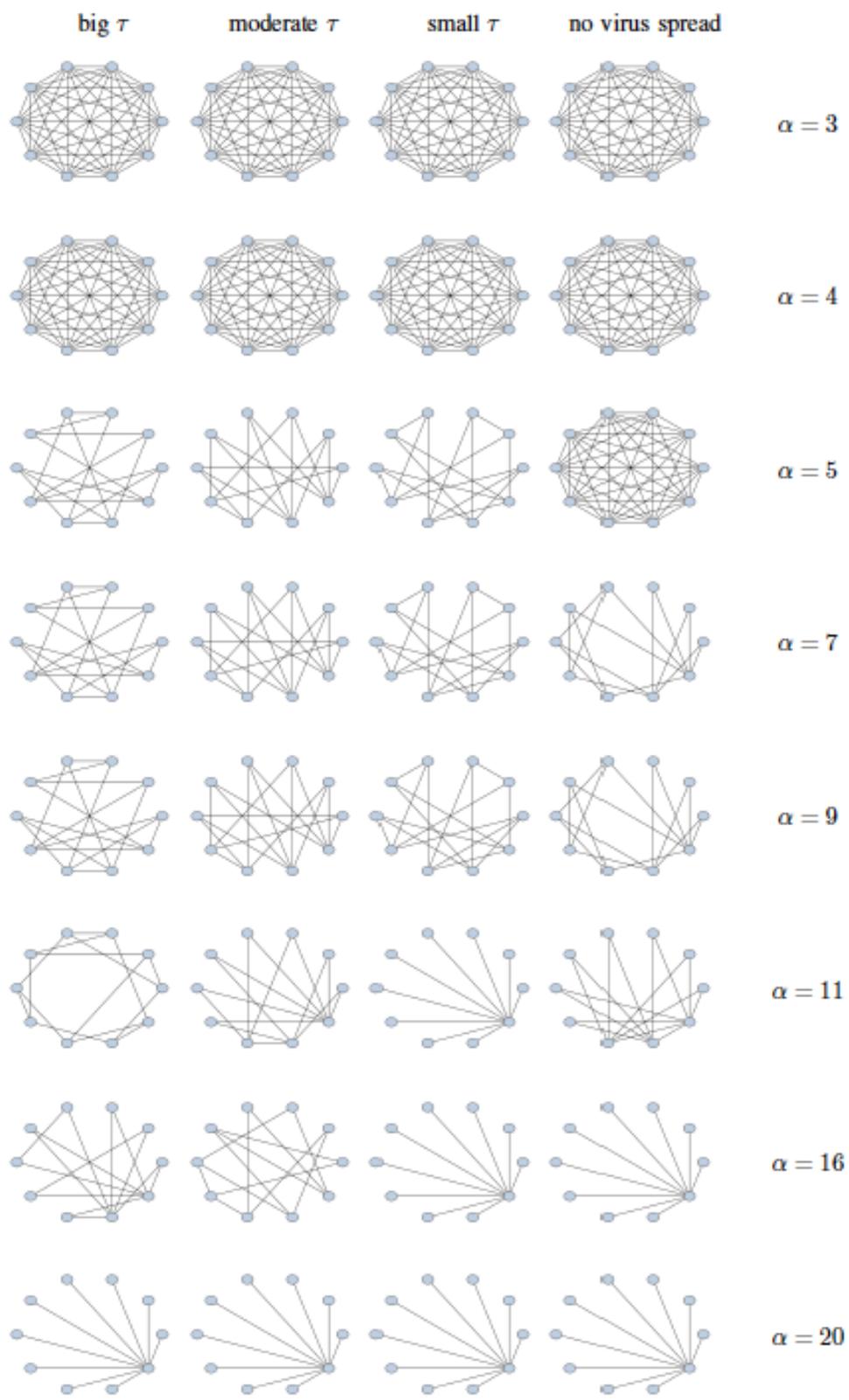}

\caption{Obtained topologies by the algorithm for $\gamma = 5$.}
 \label{fig:VizTopol2Gamma5}
\end{figure*}

\captionsetup[subfigure]{labelformat=empty}
\begin{figure*}[p]

\centering
\includegraphics[trim = 0mm 0mm 0mm 0mm,clip,width=1.0\textwidth]{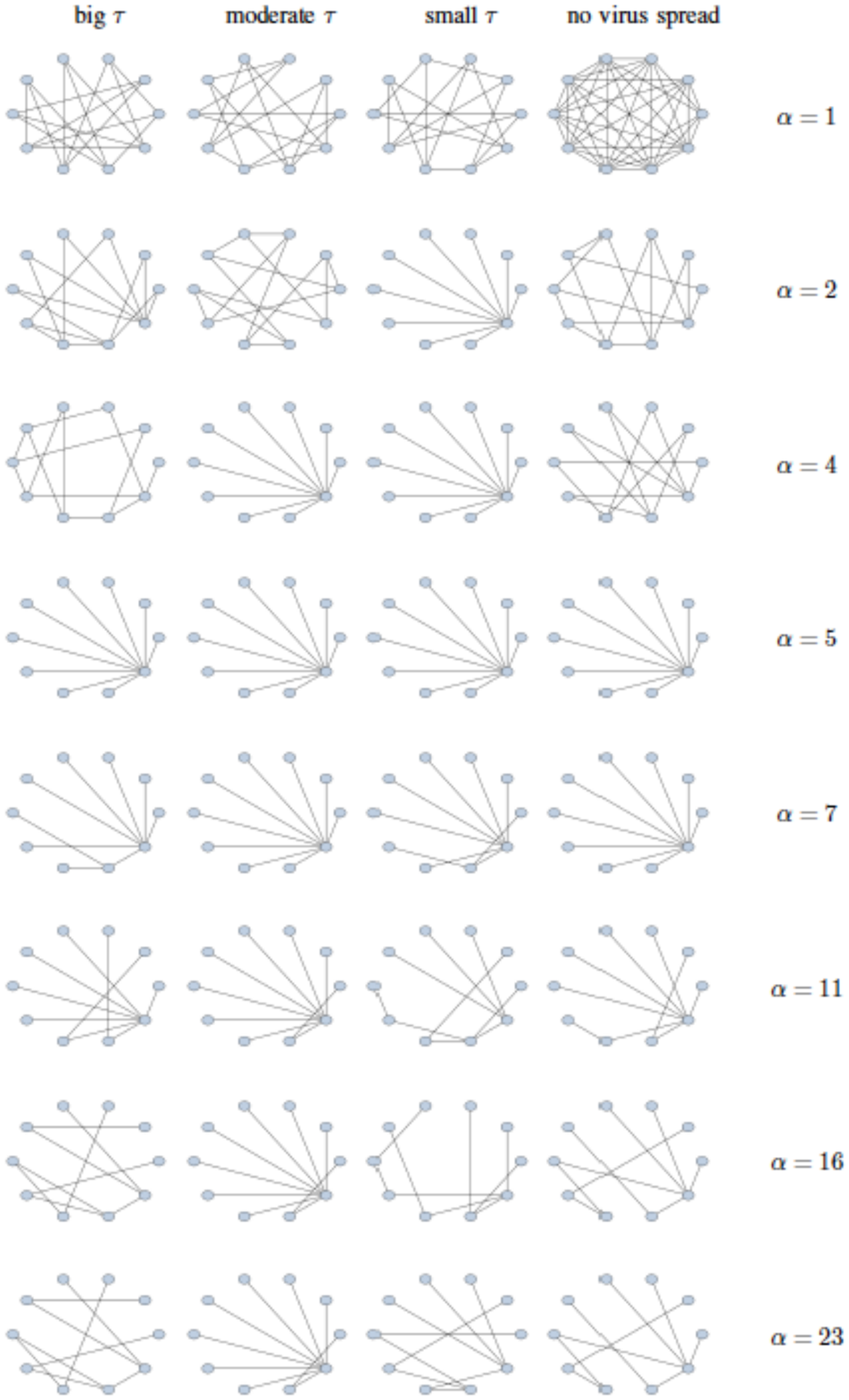}

\caption{Obtained topologies by the algorithm for $\gamma = 1$.}
 \label{fig:VizTopol2Gamma1}
\end{figure*}

\captionsetup[subfigure]{labelformat=empty}
\begin{figure*}[p]

\centering
\includegraphics[trim = 0mm 0mm 0mm 0mm,clip,width=1.0\textwidth]{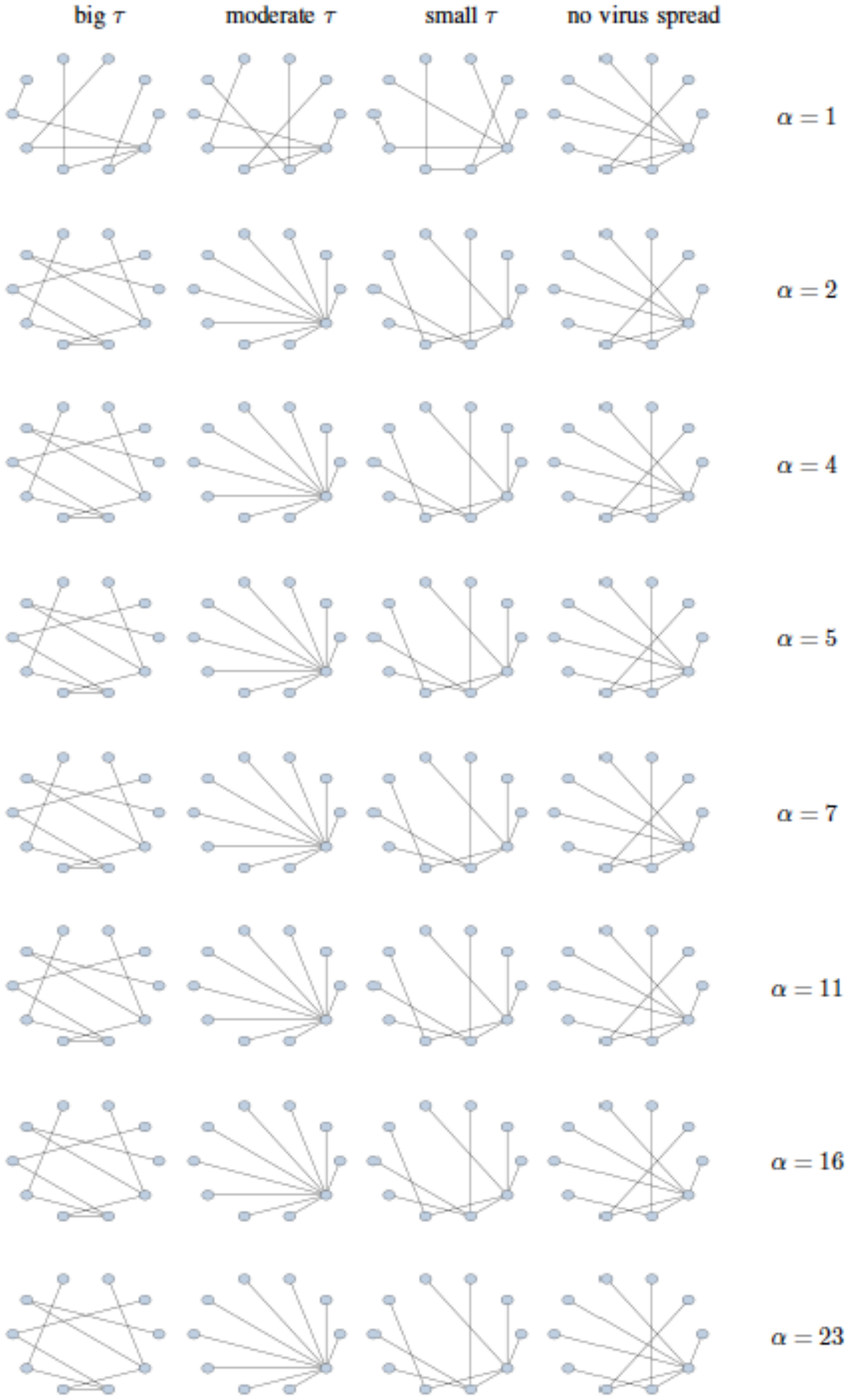}

\caption{Obtained topologies by the algorithm for $\gamma = 0.1$.}
 \label{fig:VizTopol2Gamma01}
\end{figure*}

\newpage

\linespread{0.9721}

\end{document}